\documentclass[a4paper,12pt]{article}
\pdfoutput=1 

\usepackage{jheppub} 

\usepackage[T1]{fontenc} 

\usepackage{float, array, xspace, amscd, amsthm}
\usepackage{fancyhdr}
\usepackage{longtable}

\usepackage{amsmath}
\usepackage{slashed}

\newtheorem{theorem}{Theorem}
\newtheorem{lemma}{Lemma}
\newtheorem{proposition}{Proposition}

\newtheorem{definition}{Definition}
\newtheorem{corollary}{Corollary}

\newcommand{\nn}{\nonumber}
\newcommand{\dd}{{\rm d}}
\newcommand{\w}{\wedge}
\newcommand{\End}{{\rm End}}
\newcommand{\cd}{\check\dd}
\newcommand{\tr}{{\rm tr}}
\newcommand{\dth}{{{\dd}}_\theta}
\newcommand{\cdth}{{{\check\dd}}_\theta}

\newcommand{\be}{\begin{equation}}
\newcommand{\ee}{\end{equation}}
\def\bea#1\eea{\begin{align}#1\end{align}}

\usepackage{color}

\title{The infinitesimal moduli space of heterotic $G_2$ systems
}
\author[a]{Xenia de la Ossa,}
\author[b]{Magdalena Larfors,}
\author[c,d,e]{Eirik E.~Svanes}

\affiliation[a]{Mathematical Institute, Oxford University\\Andrew Wiles Building, Woodstock Road\\Oxford OX2 6GG, UK }
\affiliation[b]{Department of Physics and Astronomy,Uppsala University\\ SE-751 20 Uppsala, Sweden}
\affiliation[c]{Sorbonne Universit\'es, CNRS, LPTHE, UPMC Paris 06, UMR 7589, 75005 Paris, France}
\affiliation[d]{Sorbonne Universit\'es, Institut Lagrange de Paris, 98 bis Bd Arago, 75014 Paris, France}

\emailAdd{delaossa@maths.ox.ac.uk, magdalena.larfors@physics.uu.se, esvanes@lpthe.jussieu.fr}

\abstract{
Heterotic string compactifications on integrable $G_2$ structure manifolds $Y$ with instanton bundles $(V,A), (TY,\tilde{\theta})$ yield supersymmetric three-dimensional vacua that are of interest in physics. In this paper, we define a covariant exterior derivative $\cal D$ and show that it  is equivalent to a heterotic $G_2$ system encoding the geometry of the heterotic string compactifications. This operator $\cal D$ acts on a bundle ${\cal Q}=T^*Y\oplus{\rm End}(V)\oplus{\rm End}(TY)$ and satisfies a nilpotency condition  $\check{\cal D}^2=0$, for an appropriate projection of $\cal D$. Furthermore,  we determine the infinitesimal moduli space of these systems and show that it corresponds to the finite-dimensional cohomology group $\check H^1_{\check{\cal D}}(\cal Q)$. We comment on the similarities and differences of our result with Atiyah's well-known analysis of deformations of holomorphic vector bundles over complex manifolds. Our analysis leads to results that are of relevance  to all orders in the $\alpha'$ expansion.
}

\begin{document}

\maketitle
\flushbottom

\newpage
\section{Introduction}

A {\it heterotic $G_2$ system} is a quadruple $([Y,\varphi], [V, A], [TY, \tilde\theta], H)$ where $Y$ is a seven dimensional manifold with an integrable $G_2$ structure $\varphi$, $V$ is a bundle on $Y$ with connection $A$, $TY$ is the tangent bundle of $Y$ with connection $\tilde\theta$, and $H$ is a three form on $Y$ determined uniquely by the $G_2$ structure.  Both connections are instanton connections, that is, they satisfy 
\[
F\wedge\psi=0\:,\qquad \tilde R\wedge\psi = 0~,
\]
where $\psi = *\varphi$, $F$ is the curvature two form of the connection $A$ on the bundle $V$, and $\tilde R$ is the curvature two form of the connection $\tilde\theta$ on $TY$.   The three form $H$ must satisfy a constraint
\[
H = \dd B + \frac{\alpha'}{4}\, (CS(A) - CS(\tilde\theta))~,\]
where $CS(A)$ and $CS(\tilde\theta)$ are the Chern-Simons forms for the connections $A$ and $\tilde\theta$ respectively, and $B$ is a two-form \footnote{Note that even though the $B$ field is called a ``two form'' is not a well defined tensor as it transforms under gauge transformations of the bundles. However, $B$ transforms in such a way that the three form $H$ is in fact well defined.}. This constraint, called the anomaly cancelation condition, mixes the geometry of $Y$ with that of the bundles. These structures have significant mathematical and physical interest. The main goal of this paper is to describe the tangent space to the moduli space of these systems.  

Determining the structure of the moduli space of supersymmetric heterotic string vacua has been an open problem since the work of Strominger and Hull \cite{Strominger:1986uh,Hull:1986kz} in 1986, in which the geometry was first described for the case of compactifications on six dimensional manifolds  with $H$-flux (Calabi--Yau compactifications without flux were first constructed by Candelas {\it et.al.} \cite{Candelas:1985en}). The geometry for the seven dimensional case was later discussed in \cite{Gunaydin:1995ku,Gauntlett:2001ur, 2001math......2142F, Firedrich:2003, Gauntlett:2003cy,Ivanov:2003nd}. Over the last 30 years very good efforts have been made to understand various aspects of the moduli of these heterotic systems. The geometric moduli space for heterotic Calabi--Yau compactifications  was determined early on \cite{Candelas:1990pi}. More recently, the infinitesimal moduli space has been determined  for heterotic Calabi--Yau compactifications with holomorphic vector bundles \cite{Anderson:2009nt,Anderson:2011ty}, and subsequently for the full Strominger--Hull system \cite{Anderson:2014xha, delaOssa:2014cia,Garcia-Fernandez:2015hja,Candelas:2016usb}.  Furthermore, the geometric moduli for $G_2$ holonomy manifolds have been determined by Joyce \cite{joyce1996:1,joyce1996:2}, and explored further in the references \cite{joyce2000,Hitchin:2000jd,Gutowski:2001fm,Beasley:2002db,Dai2003,deBoer:2005pt,2007arXiv0709.2987K,Grigorian:2008tc}. Finally, deformations of $G_2$ instanton bundles have been studied \cite{donaldson1998gauge, 2009arXiv0902.3239D, sa2009instantons, 2014SIGMA..10..083S,Ball:2016xte}.

Integrable $G_2$ geometry has features in common with even dimensional complex geometry. One can define a {\it canonical differential complex}  $\check\Lambda^*(Y)$ as a sub complex of the de Rham complex \cite{fernandez1998dolbeault}, and the associated cohomologies $\check H^*(Y)$ have similarities with the Dolbeault complex of complex geometry.  Heterotic vacua on seven dimensional non-compact manifolds with an integrable $G_2$ structure lead to four-dimensional domain wall solution that are of interest in physics \cite{Gurrieri:2004dt,deCarlos:2005kh,Gurrieri:2007jg,Held:2010az, Lukas:2010mf, Klaput:2011mz,Gray:2012md,Klaput:2012vv,Klaput:2013nla,Gemmer:2013ica,Grana:2014vxa,Maxfield:2014wea,Haupt:2014ufa,Minasian:2017eur}, and whose moduli determine the massless sector of the four-dimensional theory. Furthermore, families of $SU(3)$ structure manifolds can be studied through an embedding in integrable $G_2$ geometry. Through such embeddings, variations of complex and hermitian structures of six dimensional manifolds are put on equal footing. The $G_2$ embeddings can also be used to study flows of $SU(3)$ structure manifolds \cite{Hitchin:2000jd,chiossi,delaOssa:2014lma}.

These results from physics and mathematics prompts and paves the way for our research on the combined  infinitesimal moduli space ${\cal TM}$ of heterotic $G_2$ systems $([Y,\varphi], [V, A], [TY, \tilde\theta], H)$. This study is an extension of our work  \cite{delaOssa:2016ivz}, where we determined the combined infinitesimal moduli space ${\cal TM}_{(Y,[V,A],[TY,\tilde\theta])}$ of heterotic $G_2$ systems with $H=0$, where $Y$ is a $G_2$ holonomy manifold. The canonical cohomology for manifolds with an integrable $G_2$ structure mentioned above can be extended to bundle valued cohomologies for bundles $(V,A)$ on $Y$, as long as the connection $A$ is an instanton \cite{reyes1993some, carrion1998generalization}. As the instanton condition is the heterotic supersymmetry condition for the gauge bundle, the corresponding canonical cohomologies feature prominently in the moduli problems of heterotic compactifications.  
We find in particular, a  $G_2$ analogue of Atiyah's deformation space for holomorphic systems \cite{atiyah1957complex}. We restrict ourselves in the current paper to scenarios where the internal geometry $Y$ is compact, though we are confident that the analysis can also be applied in non-compact scenarios such as the domain wall solutions \cite{Gurrieri:2004dt,deCarlos:2005kh,Gurrieri:2007jg,Held:2010az, Lukas:2010mf, Klaput:2011mz,Gray:2012md,Klaput:2012vv,Klaput:2013nla,Gemmer:2013ica,Grana:2014vxa,Maxfield:2014wea,Haupt:2014ufa,Minasian:2017eur}, provided suitable boundary conditions are imposed.  

As a first step, we describe the infinitesimal moduli space of manifolds with an integrable $G_2$ structure.  We do this in terms of one forms with values in $TY$. On manifolds with $G_2$ holonomy, the infinitesimal moduli space of compact manifolds $Y$ \cite{joyce1996:1,joyce1996:2} is contained in $\check H^1(Y,TY)$ \cite{deBoer:2005pt,delaOssa:2016ivz} which is finite-dimensional  \cite{reyes1993some, carrion1998generalization}. For manifolds with integrable $G_2$ structure, the differential constraints on the geometric moduli are much weaker, and the infinitesimal moduli space of $Y$ need not be a finite dimensional space.  This is analogous to the infinite dimensional hermitian moduli space of the $SU(3)$ structure manifolds of the Strominger--Hull systems \cite{Becker:2005nb,Becker:2006xp}. Expressing the geometric deformations in terms of $TY$-valued one forms has another important consequence: using this formalism makes it easier to describe {\it finite deformations} of the geometry. We will use the full power of this mathematical framework  in a future publication \cite{delaOssa17} to study the finite deformation complex of integrable $G_2$ manifolds.

We then extend our work to a description of the deformations of $([Y,\varphi],[V,A])$ requiring that the instanton constraint is preserved. As mentioned above, we find a structure that resembles Atiyah's analysis of deformations of holomorphic bundles. Specifically, we find that the infinitesimal moduli space  ${\cal TM}_{([Y,\varphi],[V,A])}$ is contained in
\begin{equation*}
\check H^1(Y,{\rm End}(V))\oplus\ker(\check{\cal F})\:,
\end{equation*}
where we define a $G_2$ Atiyah map $\cal F$ by \cite{delaOssa:2016ivz}
\begin{equation*}
\check{\cal F}\::\;\;\; {\cal TM}_Y\rightarrow\check H^2(Y,{\rm End}(V))\:,
\end{equation*}
which a linear map  given in terms of the curvature $F$.  
The space ${\cal TM}_Y$ denotes the infinitesimal geometric moduli of $Y$ which, as noted above, can be infinite dimensional but reduces to $\check H^1(Y,TY)$ in the case where $Y$ has $G_2$ holonomy as showed in \cite{delaOssa:2016ivz}.

Finally we consider the full heterotic $G_2$ system, including  the heterotic anomaly cancellation equation. When combined with the instanton conditions on the bundles, we show that the constraints on the heterotic $G_2$ system $([Y,\varphi], [V,A], [TY,\tilde\theta],H)$ can be rephrased in terms of a nilpotency condition $\check{\cal D}^2=0$ on the operator $\cal D$ acting on a bundle 
\begin{equation*}
{\cal Q}=T^*Y\oplus{\rm End}(V)\oplus{\rm End}(TY)\:.
\end{equation*}
It should be noted that, in contrast to compactifications of   six dimensional complex manifolds studied in \cite{Anderson:2009nt,Anderson:2011ty,Anderson:2014xha, delaOssa:2014cia, Garcia-Fernandez:2015hja}, the operator $\check{\cal D}$ does not define ${\cal Q}$ as an extension bundle as, we will see, it is {\it not upper triangular}. 
We proceed to show that the infinitesimal heterotic moduli are elements in the cohomology group
\begin{equation*}
{\cal TM}=\check H^1_{\check{\cal D}}(\cal Q)\:.
\end{equation*}
Consequently, the infinitesimal moduli space of heterotic $G_2$ systems is of finite dimension. Our analysis complements the findings of \cite{Clarke:2016qtg}, where methods of elliptic operator theory was used to show that the infinitesimal moduli space of heterotic $G_2$ compactifications is finite dimensional when the $G_2$ geometry is compact.

The rest of this paper is organised as follows: Section \ref{sec:bg} reviews  $G_2$ structures and introduces mathematical tools we need in our analysis. Section \ref{sec:finitemod} discusses infinitesimal deformations of manifolds $Y$ with integrable $G_2$ structure. In section \ref{sec:instantons} we discuss the infinitesimal deformations of $([Y,\varphi],[V,A])$, and in section \ref{sec:InfDefHet} we deform the full heterotic $G_2$ system $([Y,\varphi],[V,A], [TY, \tilde\theta], H)$. We conclude and point out directions for further studies in section \ref{sec:concl}. Three appendices with useful formulas, curvature identities and a summary of heterotic supergravity complement the main discussion.

\section{Background material} \label{sec:bg}
This section summarises the mathematical formalism that we will need to analyse the deformations of heterotic string vacua on manifolds with $G_2$  structure. While we intend for this paper to be self-contained, we will only discuss the tools of need for the present analysis. More complete treatments can be found in  the references stated below.

\subsection{Manifolds with a $G_2$ structure} \label{sec:g2prel}

A  manifold with a $G_2$ structure is a seven dimensional manifold $Y$ which admits a non-degenerate positive associative 3-form $\varphi$ \cite{joyce2000}.  Any seven dimensional manifold which is spin and orientable, that is, its first and second Stiefel-Whitney classes are trivial,  admits a $G_2$ structure.  The 3-form $\varphi$ determines a Riemannian metric $g_\varphi$ on $Y$ given by
\begin{equation}
6 g_\varphi(x, y)\, \dd {\rm vol}_\varphi
= (x\lrcorner\varphi)\wedge(y\lrcorner\varphi)\wedge\varphi~,
\label{eq:g2metric}
\end{equation}
where $x$ and $y$ are any vectors in $\Gamma(TY)$. 
The Hodge-dual of $\varphi$ with respect to this metric is a co-associative 4-form
\[ \psi = *\varphi~.\]
The components of the metric $g_\varphi$ are
\begin{equation}
g_{\varphi\, ab} = \frac{\sqrt{\det g_\varphi}}{3!\, 4!}\, 
\varphi_{a c_1 c_2}\, \varphi_{b c_3 c_4}\, \varphi_{c_5 c_6 c_7}\,
\epsilon^{c_1\cdots c_7}
=  \frac{1}{4!}\, 
\varphi_{a c_1 c_2}\, \varphi_{b c_3 c_4}\, 
\psi^{c_1 c_2 c_3 c_4}~,
\label{eq:g2metricab}
\end{equation}
where
\[ \dd x^{a_1\cdots a_7} = \sqrt{\det g_\varphi}
\ \epsilon^{a_1\cdots a_7}\, \dd {\rm vol}_\varphi~.\]
Note that with respect to this metric, the 3-form $\varphi$, and hence its Hodge dual $\psi$, are  normalised so that
\[ \varphi\wedge *\varphi = ||\varphi||^2\, \dd{\rm vol}_\varphi
~, \qquad ||\varphi||^2 = \varphi\lrcorner\varphi= 7~.\]
We refer the reader to  \cite{MR916718,bonan66,FerGray82,Hitchin:2000jd,joyce2000,Bryant:2005mz}, and our paper \cite{delaOssa:2016ivz},  for more details on $G_2$ stuctures.

\subsubsection{Decomposition of forms}\label{sssec:formdecomp}
The existence of a $G_2$ structure $\varphi$ on $Y$ determines a decomposition of differential forms on $Y$ into irreducible representations of $G_2$.  This decomposition changes when one deforms the $G_2$ structure. 

Let $\Lambda^k(Y)$ be the space of $k$-forms on $Y$ and $\Lambda_p^k(Y)$ be the subspace of $\Lambda^k(Y)$ of $k$-forms which transform in the $p$-dimensional irreducible representation of $G_2$.   We have the following decomposition for each $k= 0,1, 2, 3$:
\begin{align*}
\Lambda^0 &= \Lambda_1^0
~,\\
\Lambda^1 &= \Lambda_7^1 = T^*Y \cong TY
~,\\
\Lambda^2 &= \Lambda_7^2\oplus \Lambda_{14}^2
~,\\
\Lambda^3 &= \Lambda_1^3\oplus\Lambda_7^3\oplus\Lambda_{27}^3
~.
\end{align*}
The decomposition for $k = 4, 5, 6, 7$ follows from the Hodge dual for $k = 3, 2, 1, 0$ respectively.

Any two form $\beta$ can be decomposed as
\[
\beta = \alpha\lrcorner\varphi + \gamma~,\]
for some $\alpha\in\Lambda^1$ and two form $\gamma\in\Lambda_{14}^2$ which satisfies $\gamma\lrcorner\varphi = 0$ (or equivalently
$\gamma\wedge\psi = 0$) where, by equations \eqref{eq:onephiphi} and \eqref{eq:twophiphi}, we have
\begin{align}
\pi_7(\beta) &= \frac{1}{3}\, (\beta\lrcorner\varphi)\lrcorner\varphi= \frac{1}{3}\, (\beta + \beta\lrcorner\psi)\label{eq:proj2to7}
~,\\
\pi_{14}(\beta) &= \frac{1}{3}\, (2\beta - \beta\lrcorner\psi)\label{eq:proj2to14}
~.
\end{align}
That is, we can characterise the decomposition of $\Lambda^2$ as follows:
\begin{align}
\Lambda_7^2 &= 
 \{ \alpha\lrcorner\varphi: \alpha\in \Lambda^1\}
 = \{\beta\in\Lambda^2: (\beta\lrcorner\varphi)\lrcorner\varphi = 3\,\beta\}
= \{\beta\in \Lambda^2: \beta\lrcorner\psi = 2\, \beta\}
~,\label{eq:2F7}\\
\Lambda_{14}^2 &= 
\{\beta\in \Lambda^2: \beta\lrcorner\varphi = 0\} 
= \{\beta\in \Lambda^2: \beta\wedge\psi = 0\}
= \{\beta\in \Lambda^2: \beta\lrcorner\psi = -\, \beta\}
~.\label{eq:2F14}
\end{align}
 The decomposition of $\Lambda^5$ is easily obtained by taking the Hodge dual of the decomposition of $\Lambda^2$, and we can write any five-form as
\[ \beta = \alpha\wedge\psi + \gamma~,\]
where $\alpha\in \Lambda^1$,  and $\gamma\in\Lambda_{14}^5$ satisfies $\psi\lrcorner\gamma = 0$.  The decomposition of $\Lambda^5$ are then analogous to \eqref{eq:2F7}-\eqref{eq:2F14}, and can be found in \cite{delaOssa:2016ivz}. An alternative representation of five-forms is
\[ \beta = \alpha\wedge\psi + \varphi\wedge\sigma~,\]
where $\sigma\in\Lambda_{14}^2$ and $*\gamma = -\sigma$.    The components $\alpha$ and $\sigma$ can be obtained by performing the appropriate contractions with $\psi$ or $\varphi$ respectively
\begin{equation*}
\alpha = \frac{1}{3}\, \psi\lrcorner\beta~,\qquad
\sigma = \varphi\lrcorner\beta - \frac{2}{3}\, (\psi\lrcorner\beta)\lrcorner\varphi~.
\end{equation*}

Any three form $\lambda$ can be decomposed into
\begin{equation}
 \lambda = f\, \varphi + \alpha\lrcorner\psi + \chi~,
 \label{eq:decomp3}
 \end{equation}
for some function $f$, some $\alpha\in\Lambda^1$, and some three form $\chi\in\Lambda^3_{27}$ which satisfies
\[ \chi\lrcorner\varphi = 0~,\qquad
{\rm and}\qquad \chi\lrcorner\psi = 0~.\]
Another way to characterise and decompose a three form is in terms of a one form $M$ with values in the tangent bundle. Given such form $M\in\Lambda^1(TY)$, there is a unique three form
\begin{equation}
\lambda = \frac{1}{2}\, M^a\wedge\varphi_{abc}\, \dd x^{bc}~.\label{eq:matrix3}
\end{equation}
Conversely, a three form $\lambda$ determines a unique one from $M\in \Lambda^1(TY)$ 
\begin{align}
\frac{1}{4}\,  \varphi^{cd}{}_a\, \lambda_{bcd} &= \frac{1}{2}\,  g_{ab}\, {\rm tr}M +  M_{ab}
+ \frac{1}{2}\,  M_{cd}\, \psi^{cd}{}_{ab}
\nn\\[5pt]
& = \frac{9}{14}\, g_{ab}\, {\rm tr}M + h_{ab} + 3\, (\pi_7 (m))_{ab}~, \label{eq:decomp3bis}
\end{align}
where the matrix $M_{ab}$ is defined as 
\[ M_{ab} = g_{ac}\, (M^c)_b~,\]
and we have set
\be
 h_{ab} = M_{(ab)} - \frac{1}{7}\, g_{ab}\, {\rm tr} M~,
 \qquad\qquad m = \frac{1}{2}\, M_{[ab]} \, \dd x^{ab}~.\label{eq:handm}
 \ee
Comparing the decompositions \eqref{eq:matrix3} and  \eqref{eq:decomp3} we have
 \begin{align}
 f &=  \frac{3}{7}\, {\rm tr}M = \frac{1}{7}\, \varphi\lrcorner\lambda~,
 \label{eq:three1}\\[5pt]
 \alpha &=  - m\lrcorner\varphi~,
 \qquad\qquad \pi_7(m) = - \frac{1}{3}\, \alpha\lrcorner\varphi
 =  \frac{1}{4!}\, \varphi^{cd}{}_a\, \lambda_{bcd}\, \dd x^{ab}~,
 \label{eq:three7}\\[5pt]
 \chi &= \frac{1}{2}\, h_a^d\, \varphi_{bcd}\, \dd x^{abc}~,
 \qquad\qquad h_{ab} = \frac{1}{4}\, \varphi^{cd}{}_{(a}\, \chi_{b)cd}~.
 \label{eq:three27}
 \end{align}
 In other words, regarding $M$ as a matrix, $\pi_1(\lambda)$ corresponds to the trace of $M$, $\pi_7(\lambda)$
corresponds to $\pi_7(m)$ where $m$ is the antisymmetric part of $M$, and the elements in $\Lambda^3_{27}$ to the traceless symmetric 2-tensor $h_{ab}$.  It is in fact easy to check that $\chi\in\Lambda^3_{27}$ as  $\chi\lrcorner\psi= 0$ due to the symmetric property of $h$, and $\varphi\lrcorner\chi = 0$ due to $h$ being traceless.

 The decomposition of four forms can be obtained similarly.   Any four form $\Lambda$ decomposes into 
\begin{equation}
 \Lambda = \tilde f\, \psi + \tilde\alpha\wedge\varphi + \gamma~.
 \label{eq:fourdecomp}
 \end{equation}
where $\tilde f$ is a smooth function on $Y$, $\tilde\alpha$ is a one-form, and $\gamma\in\Lambda^4_{27}$ which means 
 $\varphi\lrcorner\gamma = 0$ and  $\psi\lrcorner\gamma = 0$.  We can also characterise and decompose four forms in terms of a one form $N$ with values in the tangent bundle
\begin{equation}
 \Lambda = \frac{1}{3!}\, N^a\wedge\psi_{abcd}\, \dd x^{bcd}~.
 \label{eq:fourtoM}
 \end{equation}
In this case
\[- \frac{1}{12}\, \psi^{cde}{}_a\,\Lambda_{bcde}
= \frac{8}{7}\, g_{ab}\, {\rm tr}N + S_{ab}  + 3 (\pi_7(n))_{ab}~,\]
where
\[
 S_{ab} = N_{(ab)} - \frac{1}{7}\, g_{ab}\, {\rm tr} N~,
 \qquad\qquad n = \frac{1}{2}\, N_{[ab]} \, \dd x^{ab}~.\]
The decomposition of the four form $\Lambda$ into irreducible representations of $G_2$,  is given  in terms of $N$ by
 \begin{align}
 \tilde f &=  \frac{4}{7}\, {\rm tr}N = \frac{1}{7}\, \psi\lrcorner\Lambda
 \label{eq:four1}\\[5pt]
 \tilde\alpha &=   n\lrcorner\varphi~,
 \qquad\qquad \pi_7(n) =   \frac{1}{3}\, \tilde\alpha\lrcorner\varphi
 = -  \frac{1}{3\cdot 4!}\, \psi^{cde}{}_a\, \Lambda_{bcde}\, \dd x^{ab}
 \label{eq:four7}\\[5pt]
 \gamma &= \frac{1}{3!}\, h_a^e\, \psi_{ebcd}\, \dd x^{abcd}~,
 \qquad\qquad h_{ab} = - \frac{1}{12}\, \psi^{cde}{}_{(a}\, \gamma_{b)cde}~.
 \label{eq:four27}
 \end{align}
It is easy to check that, in fact, $\gamma\in\Lambda^4_{27}$, as $\varphi\lrcorner\gamma = 0$ due to the symmetric property of $h$, and $\psi\lrcorner\gamma = 0$ due to $h$ being traceless.  Of course, this characterisation and decomposition of four forms can also be obtained using Hodge duality.  
Note also that if $\gamma\in\Lambda^4_{27}$ is given by $\gamma = *\chi$ where $\chi\in\Lambda^3_{27}$, then for
\[ \chi = \frac{1}{2}\, h_a^d\, \varphi_{bcd}\, \dd x^{abc}~,\]
we have
\[ \gamma = *\chi = - \frac{1}{3!}\, h_a^e\, \psi_{ebcd}\, \dd x^{abcd}~.\]

We will use these characterisations of three and four forms in terms of one forms with values in $TY$ to describe deformations of the $G_2$ structure, in particular, the deformations of the $G_2$ forms $\varphi$ and $\psi$.  It is important to keep in mind that only $\pi_7(m)$ and $\pi_7(n)$ appear in these decompositions.  In fact, we have not set $\pi_{14}(m)$ or $\pi_{14}(n)$ to zero as these automatically drop out.    Later, when extending our discussion of the moduli space of heterotic string compactifications, the components
$\pi_{14}(m)$ or $\pi_{14}(n)$ will enter in relation to deformations of the $B$-field.

\subsubsection{The intrinsic torsion}
\label{sec:inttor}

Decomposing into representations of $G_2$ the exterior derivatives of $\varphi$ and $\psi$ we have
\begin{align}
\dd\varphi &= \tau_0\psi + 3\, \tau_1\wedge\varphi + *\tau_3~,
\label{eq:Intphi}\\
\dd\psi &= 4\, \tau_1\wedge\psi + *\tau_2~,\label{eq:Intpsi}
\end{align}
where the forms $\tau_i\in \Lambda^i(Y)$ are called the {\it torsion classes}.  These forms are {\it uniquely} determined by the $G_2$-structure 
$\varphi$ on $Y$ \cite{FerGray82}. We note that $\tau_2\in \Lambda^2_{14}$ and that $\tau_3\in \Lambda^3_{27}$. 
A $G_2$ structure for which 
\[ 
\tau_2 = 0~,
\]
will be called an {\it integrable} $G_2$ structure following Fern\'andez-Ugarte \cite{fernandez1998dolbeault}. In this paper we will derive some results for manifolds with a general $G_2$ structure, however we will be primarily interested in integrable $G_2$ structures which are particularily relevant for heterotic strings compactifications.

We can write equations \eqref{eq:Intphi} and \eqref{eq:Intpsi} in terms of $\tau_2$ and a three form $H$ defined as
\begin{equation}
H = \frac{1}{6}\, \tau_0\, \varphi - \tau_1\lrcorner\psi - \tau_3~.\label{eq:H}
\end{equation}
In fact, one can prove that
\begin{align}
\dd \varphi &= \frac{1}{4}\, H_{ab}{}^e\, \varphi_{ecd}\, \dd x^{abcd}~,
\label{eq:dphitorsion}
\\[5pt]
\dd\psi &=  \frac{1}{12}\, H_{ab}{}^f\, \psi_{fcde}\, \dd x^{abcde} + *\tau_2~.
\label{eq:dpsitorsion}
\end{align}
The proof is straightforward using identities \eqref{eq:g2phiep}, \eqref{eq:star327one}, \eqref{eq:onephipsi} and \eqref{eq:star327two}.

Let us end this discussion with a remark on the connections on $Y$.  Let $Y$ be a manifold which has a $G_2$ structure $\varphi$, and let  $\nabla$ be a metric connection on $Y$ compatible with the $G_2$ structure, that is
\[ \nabla g_\varphi=0~, \qquad \nabla\varphi = 0~.\]
We say that the connection $\nabla$ has $G_2$ holonomy.   The conditions $\nabla\varphi=0$ and $\nabla\psi = 0$ imply equations  \eqref{eq:dphitorsion} and \eqref{eq:dpsitorsion} respectively, and the three form  $H$ corresponds to the torsion of the unique connection which is totally antisymmetric which exists {\it only if} $\tau_2=0$ \cite{Bryant:2005mz}.

\subsubsection{The canonical cohomology}\label{sec:cancom}
Before we go on, we need to introduce the concept of a ``Dolbeault complex'' for manifolds with an integrable $G_2$ structure. This complex is appears naturally in the analysis of  infinitesimal and finite deformations of integrable $G_2$ manifolds and heterotic compactifications. It was first considered in \cite{reyes1993some,fernandez1998dolbeault}, and discussed extensively in \cite{delaOssa:2016ivz}, so we will limit our discussion to  the necessary definitions and theorems. In the ensuing sections, we will use and generalise these results.

To construct a sub-complex of the de Rham complex of $Y$, we define the analogue of a Dolbeault operator on a complex manifold
\begin{definition} The differential operator $\cd$ is defined by 
the maps 
\begin{align*}
\cd_0&: \Lambda^0(Y)\rightarrow\Lambda^1(Y)~,\qquad\qquad  
\cd_0f = \dd f~, \qquad\quad f\in\Lambda^0(Y)
~,\\
\cd_1&: \Lambda^1(Y)\rightarrow\Lambda_7^2(Y)~,\qquad\qquad 
\cd_1\alpha = \pi_7(\dd \alpha)
~, \quad \alpha\in\Lambda^1(Y)
~,\\
\cd_2 &: \Lambda_7^2(Y)\rightarrow\Lambda_1^3(Y)~,\qquad\qquad
\cd_2\beta = \pi_1(\dd\beta)~, \quad \beta\in\Lambda_7^2(Y)~.
\end{align*}
That is, 
\begin{equation*}
\cd_0=\dd~,\quad\cd_1=\pi_7\circ\dd~,\quad\cd_2=\pi_1\circ\dd ~.
\end{equation*}
\end{definition}
Then we have the following theorem \cite{reyes1993some, fernandez1998dolbeault}
\begin{theorem}\label{prop:dolbcomplex}
Let $Y$ be a manifold with a $G_2$ structure. Then
\begin{equation}
\label{eq:dolb}
0\rightarrow\Lambda^0(Y)\xrightarrow{\cd}\Lambda^1(Y)\xrightarrow{\cd}\Lambda^2_7(Y)\xrightarrow{\cd}\Lambda^3_1(Y)\rightarrow0
\end{equation}
is a differential complex, i.e. $\cd^2=0$ if and only if the $G_2$ structure is integrable, that is, $\tau_2 = 0$ .
\end{theorem}
We denote the complex \eqref{eq:dolb} by $\check\Lambda^*(Y)$. This complex \eqref{eq:dolb} is, in fact, an elliptic complex \cite{reyes1993some}. The corresponding cohomology ring, $\check H^*(Y)$, is referred to as the canonical $G_2$-cohomology of $Y$ \cite{fernandez1998dolbeault}.

This complex can naturally be extended to forms with values in bundles, just as for holomorphic bundles over a complex manifold. Let $E$ be a bundle over the manifold $Y$ with a one-form connection $A$ whose curvature is $F$. We are interested in instanton connections $A$ on $E$, that is, connections with curvature $F$ which satisfies
\begin{equation}
\label{eq:holbund}
\psi\wedge F=0\:,
\end{equation}
or equivalently, $F\in\Lambda^2_{14}(Y,\End(E))$. We can now define the differential operator
\begin{definition}\label{def:checkdA}
The maps $\cd_{iA}, i=0,1,2$  are given by
\begin{align*}
\cd_{0A}&: \Lambda^0(Y,E)\rightarrow\Lambda^1(Y,E)~,\qquad  
\qquad 
\cd_{0A}f = \dd_A f~, \qquad\quad f\in\Lambda^0(Y,E)
~,\\
\cd_{1A}&: \Lambda^1(Y,E)\rightarrow\Lambda^2_7(Y,E)~,\qquad\qquad 
\cd_{1A}\alpha = \pi_7(\dd_A \alpha)
~, \quad \alpha\in\Lambda^1(Y,E)
~,\\
\cd_{2A} &: \Lambda^2(Y,E)\rightarrow\Lambda^3_1(Y,E)~,\qquad\qquad
\cd_{2A}\beta = \pi_1(\dd_A\beta)~, \quad \beta\in\Lambda_7^2(Y,E)~.
\end{align*}
where the $\pi_i$'s denote projections onto the corresponding subspace.
\end{definition}
It is easy to see that these operators are well-defined under gauge transformations. Theorem \ref{prop:dolbcomplex} can then be generalised to \cite{reyes1993some}:
\begin{theorem}
Let $Y$ be a seven dimensional manifold with a $G_2$-structure. The complex
\begin{equation}
\label{eq:dolbV}
0\rightarrow\Lambda^0(Y,E)\xrightarrow{\cd_A}\Lambda^1(Y,E)\xrightarrow{\cd_A}\Lambda^2_7(Y,E)\xrightarrow{\cd_A}\Lambda^3_1(Y,E)\rightarrow0
\end{equation}
is a differential complex, i.e. $\cd_A^2=0$, if and only if the connection $A$ on $V$ is an instanton and the manifold has an integrable $G_2$ structure. We shall denote the complex \eqref{eq:dolbV} $\check\Lambda^*(Y,E)$.
\end{theorem}
Note that that the complex \eqref{eq:dolbV} is elliptic, as was shown in \cite{carrion1998generalization}.

\subsection{Useful tools for deformation problems}

In this section, we review and develop tools for the study of the moduli space of (integrable) $G_2$ structures. 
While the ulterior motive to introduce this mathematical machinery is to investigate whether the moduli space of heterotic string compactifications is given by a differential graded Lie Algebra (DGLA), we limit ourselves in this paper to infinitesimal deformations. A more thorough discussion about DGLAs and finite deformations will appear elsewhere \cite{delaOssa17}. For more discussion about the graded derivations, insertion operators and derivatives introduced below, the reader is referred to e.g.~\cite{huybrechts,manetti,kont-soib}.

\subsubsection{Graded derivations and insertion operators}

Let $Y$ be a manifold of arbitrary dimension.

\begin{definition}
A graded derivation $D$ of degree $p$ on a manifold $Y$  is a linear map
\begin{equation*}
D \, : \qquad \Lambda^k(Y) \longrightarrow \Lambda^{p+ k}(Y)~,
\end{equation*}
which satisfies the Leibnitz rule
\begin{equation}
D(\alpha\wedge\beta) = D(\alpha)\wedge\beta + (-1)^{kp}\, \alpha\wedge D(\beta)~.
\label{eq:leibnitz}
\end{equation}
for all $k$-forms $\alpha$ and any form $\beta$.

\end{definition}

\begin{definition}
Let M be a $p$-form with values in $TY$ and let $\alpha$ be a $k$-form.  The  insertion operator $i_M$ is defined by the linear map
\begin{align}
i_M \, : \quad \Lambda^k(Y) &\longrightarrow \Lambda^{p+ k -1}(Y)~,
\nn\\
\alpha &\longmapsto
i_M(\alpha) = \frac{1}{ (k-1)!}~ M^a\wedge\alpha_{ab_1\cdots b_{k-1}}\, \dd x^{b_1\cdots b_{k-1}} = M^a\wedge\alpha_a~,\label{eq:insop}
\end{align}
where we have defined a $(k-1)$ form $\alpha_a$ with values in $T^*Y$ from the $k$-form $\alpha$ by
\[ \alpha_a = \frac{1}{ (k-1)!}~ \alpha_{ab_1\cdots b_{k-1}}\, \dd x^{b_1\cdots b_{k-1}}~.\] 

\end{definition}

\noindent It is not too hard to prove that the insertion operator $i_M$ defines a graded derivation of degree $p-1$, and we leave this as an exercise for the reader.
 
One can extend the definition of the insertion operator to act on the space of forms with values in $\Lambda^n\,TY$, or $\Lambda^n\,T^*Y$, or indeed in $\Lambda^n\,V\times\Lambda^m V^*$, for any bundle $V$ on $Y$. 
For forms with values in any bundle $E$ on $Y$, the insertion operator $i_M$ is the linear map
\begin{align}
i_M \, : \quad \Lambda^k(Y, E) &\longrightarrow \Lambda^{p+ k -1}(Y, E)~,
\label{eq:insopE}
\end{align}
with $i_M(\alpha)$ given by the same formula
\eqref{eq:insop} for any $\alpha\in\Lambda^k(E)$.  Again, it is not too hard to see that this formula defines a graded derivation of degree $p-1$. For example, for every $M\in\Lambda^p(Y, TY)$ and $N\in\Lambda^q(T, TY)$ we define $i_M(N)\in\Lambda^{p+q - 1}(Y,TY)$ by 
\begin{equation}
i_M(N^a) = \frac{1}{(q-1)!}\, M^b\wedge (N^a)_{bc_1\cdots c_{q-1}}
\, \dd x^{c_1\cdots c_{q-1}}
~.\label{eq:extins}
\end{equation}
A further generalisation can be achieved by letting the form $M$ which is being inserted take values in $\Lambda^p(\Lambda^m TY)$ for $m\ge 1$.  For example, the insertion operator $i_M$ for the action of $M\in\Lambda^p(\Lambda^m TY)$ on $N\in\Lambda^q(Y, TY)$ is given by
\[ i_M(N) = M^{a_1\cdots a_m}\wedge N_{a_1\cdots a_m}~,\]
where  $q\ge m$ and
\[ N_{a_1\cdots a_m} = \frac{1}{(q-m)!}\, N_{a_1\cdots a_m b_1\cdots b_{q-m}}
\dd x^{b_1\cdots b_{q-m}}~.\]
In this case, $i_M$ is a derivation of degree $p-m$.

The insertion operators $i_M$ form a Lie algebra with a bracket $[\cdot, \cdot]$ given by
\begin{equation}
[ i_{M}, i_{N} ] = i_M i_N - (-1)^{(p-1)(q-1)}\, i_N i_M
=i_{[M, N]}~, \label{eq:insbracket}
\end{equation}
where $M\in\Lambda^p(Y, TY)$, $N\in\Lambda^q(T, TY)$ and 
\begin{equation}
[M, N]= i_M(N) - (-1)^{(p-1)(q-1)}\,i_N(M)~, \label{eq:NijRich}
\end{equation}
is the Nijenhuis-Richardson bracket, which is a derivation of degree $p+q -1$. The Lie bracket is a derivation of degree $p+q -2$.  To verify \eqref{eq:insbracket}, let $\alpha$ be any $k$-form, (perhaps with values in a bundle $E$ on $Y$).  Then, by the Leibnitz rule \eqref{eq:leibnitz}
\begin{align}
i_M(i_N(\alpha)) &= i_M(N^a\wedge\alpha_a)
= i_M(N^a)\wedge\alpha_a + (-1)^{(p-1)q}\, N^a\wedge i_M(\alpha_a)
\nn\\
&= i_{i_M(N)}(\alpha) + (-1)^{(p-1)q}\, N^a\wedge M^b\wedge\alpha_{ab}~,
\label{eq:doubleins}
\end{align}
where $\alpha_{ab}$ is the $(k-2)$-form obtained from $\alpha$
\[ \alpha_{ab}= \frac{1}{ (k-2)!}~ \alpha_{abc_1\cdots c_{k-2}}\, \dd x^{c_1\cdots c_{k-1}}~.\]
Then noting equation \eqref{eq:NijRich} and that
\[ M^a\wedge N^b\wedge\alpha_{ab} = 
(-1)^{pq + 1}\, N^a\wedge M^b\wedge\alpha_{ab}~,\]
we obtain \eqref{eq:insbracket}.

\begin{definition}
The Nijenhuis-Lie derivative ${\cal L}_M$ along
$M\in\Lambda^p(Y, TY)$ is defined by
\begin{equation}
 {\cal L}_M = [\dd, i_M] = \dd\, i_M + (-1)^p\, i_M\, \dd~,
 \label{eq:NiejLie}
 \end{equation}
where $\dd$ is the exterior derivative.  
\end{definition}
{\noindent}Note that when $p=1$, $M$ is a section of $TY$ and so the Nijenhuis-Lie derivative is the Lie derivative along the vector field $M$. 
The Nijenhuis-Lie derivative is a derivation of degree $p$ acting on the space of forms on $Y$.  

\subsubsection{Covariant derivatives, connections and Lie derivatives}
\label{sec:covder}

We can generalise the definition of the Nijenhuis-Lie to act covariantly on forms with values in any bundle $E$. This was also recently discussed in \cite{2014arXiv1412.2533D}. Suppose that $\alpha$ is $k$-form on $Y$  which transforms in a representation of the gauge group of $E$ with representation matrices $T_I$, where the label $I$ runs over the dimension of the gauge group.   
Then, an exterior covariant derivative we can be written as
\begin{equation}
\dd_A\, \alpha = \dd\, \alpha + A\cdot \alpha~,
\qquad A\cdot \alpha = A^I\, \wedge (T_I\, \alpha)~.
\label{eq:gendA}
\end{equation}
where $A$ is a connection one form on $E$.   Note that
\[ \dd_A^2\alpha = F\cdot\alpha~,\]
where $F$ is the curvature of the connection $A$.

\begin{definition}
Let $E$ be a vector bundle on $Y$ with connection $A$. The covariant Nijenhuis-Lie derivative ${\cal L}^A_M$ along
$M\in\Lambda^p(Y, TY)$ acting on forms on $Y$ which are in a representation of the $E$ is defined by
\begin{equation}
 {\cal L}^A_M = [\dd_A, i_M] = \dd_A\circ i_M + (-1)^p\, i_M\circ \dd_A~.
 \label{eq:CovNiejLie}
 \end{equation}
\end{definition}

Let $\nabla$ be a covariant derivative  on $Y$ with connection symbols $\Gamma$. One can define a covariant derivative $\nabla^A$ on $E\otimes TY$ (to make sense of parallel transport on $E$) by 
\begin{equation}
\nabla_a^A\, \alpha_{c_1\cdots c_k}
= \partial_{A\, a} \, \alpha_{c_1\cdots c_k}
- k\, \Gamma_{a[c_1}{}^b\, \alpha_{|b|c_2\cdots c_k]}
=\nabla_a\, \alpha_{c_1\cdots c_k} + A_a\cdot\alpha_{c_1\cdots c_k}
~,
\label{eq:nablaA}
\end{equation}
where 
\[
\partial_{A\, a} \, \alpha_{c_1\cdots c_k}
= \partial_a \, \alpha_{c_1\cdots c_k}
+ A_a\cdot\alpha_{c_1\cdots c_k}~,
\]

Let $\dth$ be an exterior covariant derivative on $TY$ with connection one form $\theta$ given by
\begin{equation}
 \theta_a{}^b = \Gamma_{ac}{}^b\, \dd x^c~,\label{eq:theta}
\end{equation}
where $\Gamma$ are the connection symbols of a covariant derivative $\nabla$ on $Y$.

\begin{theorem}\label{theorem:CNL}
Let $E$ be a bundle on a manifold $Y$ with connection $A$. 
The covariant Nijenhuis-Lie derivative ${\cal L}^A_M$ along $M\in\Lambda^p(Y, TY)$
satisfies
\begin{equation}
{\cal L}_M^A = [\dd_A, i_M] =  i_{\dth\, M} + (-1)^p\, i_M(\nabla^A)~,
\label{eq:CNL}
\end{equation}
where $\dth$ is an exterior covariant derivative on $TY$ with connection one form
\[
\theta_a{}^b = \Gamma_{ac}{}^b\, \dd x^c~,
\]
and $\nabla^A$ is a covariant derivative on $E\otimes TY$ with connection symbols on $TY$ given by $\Gamma$.
\end{theorem}

\begin{proof}
Let $\alpha$ be any $k$-form on $Y$ which transforms in a representation of the structure group of $E$ with representation matrices $T_I$.  Then
\begin{align*}
\dd_A\, i_M(\alpha) &= \dd_A(M^a\wedge \alpha_a)
= \dth M^a\wedge \alpha_a + (-1)^p\, M^a\wedge (
\dd_A\alpha_a - \theta_a{}^b\, \wedge\alpha_a )
\\
& = i_{\dth M}(\alpha) - (-1)^p\, i_M(\dd_A\alpha)
+ (-1)^p\, M^a\wedge (
\dd_A\alpha_a + (\dd_A\alpha)_a - \theta_a{}^b\, \wedge\alpha_b )~.
\end{align*}
For the third term we have
\begin{align*}
\dd_A\alpha_a
&= \frac{1}{(k-1)!}\, \partial_{A\, b}\, \alpha_{ac_1\dots c_{k-1}}\, 
\dd x^{b c_1\dots c_{k-1}}
\\[5pt]
&= \frac{1}{k!}\, \big(
(k+1)\, \partial_{A[b}\alpha_{ac_1\cdots c_{k-1}]}
+ (-1)^{k-1} \,\partial_{A\, a}\, \alpha_{c_1\cdots c_{k-1}b}
\big)\, \dd x^{b c_1\dots c_{k-1}}
\\[5pt]
&= \frac{1}{k!}\, (\dd_A\alpha)_{bac_1\cdots c_{k-1}}\, \, \dd x^{b c_1\dots c_{k-1}} + \partial_{A\, a}\alpha
\\[5pt]
&= - (\dd_A\alpha)_a + \partial_{A\, a}\alpha~.
\end{align*}
Therefore
\begin{equation}\label{eq:covid}
\dd_A\alpha_a + (\dd_A\alpha)_a - \theta_a{}^b\, \wedge\alpha_b
=  \partial_{A\, a}\alpha  - \theta_a{}^b\, \wedge\alpha_b~.
\end{equation}
This result can be written in terms of a gauge covariant derivative $\nabla^A$ on $E\otimes TY$
\[\partial_{A\, a}\alpha  - \theta_a{}^b\, \wedge\alpha_b
=  \frac{1}{k!}\, (\partial_{A\, a}\alpha_{c_1\cdots c_k}
- k\, \Gamma_{a[c_1}{}^b\, \alpha_{|b|c_2\cdots c_k]})\, \dd x^{c_1\cdots c_k}
=  \frac{1}{k!}\, (\nabla_a^A\, \alpha_{c_1\cdots c_k})\, \dd x^{c_1\cdots c_k}~.
\]
Thus
\begin{equation*}
\dd_A\, i_M(\alpha) =i_{\dth M}(\alpha) - (-1)^p\, i_M(\dd_A\alpha)
+ (-1)^p\, M^a\wedge \nabla_a^A\alpha~,
\end{equation*}
and \eqref{eq:CNL} follows.

\end{proof}

Note the useful expression in the proof for the covariant derivative, namely
\be
\nabla_a^A\,\alpha \equiv \frac{1}{k!}\, (\nabla_a^A\, \alpha_{c_1\cdots c_k})\, \dd x^{c_1\cdots c_k}= \partial_{A\, a}\alpha  - \theta_a{}^b\, \wedge\alpha_b
~.\label{eq:cool}
\ee

\begin{corollary}\label{cor:lielambda}
Let $Y$ be a n-dimensional manifold.    Let $\nabla$ be a metric compatible covariant derivative  on $Y$
with connection symbols $\Gamma$,  and $\dth$ be an exterior covariant derivative  on $TY$ such that the connection one forms $\theta$ and the connection symbols $\Gamma$ are related by 
\[ \theta_a{}^b = \Gamma_{ac}{}^b\, \dd x^c~.\]  
Suppose that $Y$ admits a $k$-form $\lambda$ which is covariantly constant with respect to $\nabla$.  
Then
\begin{equation*}
{\cal L}_M(\lambda) = [\dd, i_M](\lambda) =  i_{\dth\, M}(\lambda) ~,\label{eq:Liephi}
\end{equation*}
\end{corollary}
\begin{proof}
This follows directly from the theorem.

\end{proof}

\noindent It is important to notice that the choice for $\Gamma$ and hence $\theta$ is determined by the fact that $\nabla\lambda = 0$. Note that the Nijenhuis-Lie derivative is defined with no reference to any covariant derivate on $Y$, that is, it should only depend on the intrinsic geometry of $Y$.

\subsection{Application to manifolds with a $G_2$ structure}\label{ssec:applyG2}

Before embarking on the analysis of moduli spaces, we apply some of the ideas in the previous section to seven dimensional manifolds $Y$ with a $G_2$ structure $\varphi$.  

Let $\hat H\in\Lambda^2(Y, TY)$ be defined in terms of the three form $H$ in equation \eqref{eq:H} as
\begin{equation}
\hat H^a = \frac{1}{2}\, H_{bc}{}^a\, \dd x^{bc}~.\label{eq:hatH}
\end{equation}
Then, the integrability equations for $\varphi$ and $\psi$ in equations \eqref{eq:dphitorsion} and \eqref{eq:dpsitorsion} can be nicely written in term of insertion operators as
\begin{align}
\dd \varphi &= \hat H^a\wedge\varphi_a = i_{\hat H}(\varphi)~,
\label{eq:dphi}
\\[5pt]
\dd\psi &=\hat H^a\wedge\psi_a =  i_{\hat H}(\psi) ~,
\label{eq:dpsi}
\end{align}
where we have set $\tau_2=0$ as we are interested on moduli spaces of integrable $G_2$ structures. 

Let $\nabla$ be a covariant derivative on $Y$ compatible with the $G_2$ structure, that is
\[ \nabla\varphi = 0~,\qquad\nabla\psi = 0~,\]
with connection symbols $\Gamma$.
Then, by corollary \ref{cor:lielambda}, the Nijenhuis-Lie derivatives of $\varphi$ and $\psi$ along $M\in\Lambda^p(Y, TY)$ are
\begin{align}
{\cal L}_M(\varphi) &= [\dd, i_M](\varphi) =  i_{\dth\, M}(\varphi) ~,\label{eq:Liephi}
\\
{\cal L}_M(\psi) &=[\dd, i_M](\psi) =  i_{\dth\, M}(\psi) ~,\label{eq:Liepsi}
\end{align}
where the connection one-form $\theta$ of exterior covariant derivative $\dth$ on $TY$  is 
\[ \theta_a{}^b = \Gamma_{ac}{}^b\, \dd x^c~.\]
As mentioned before, though these equations seem to depend on a choice of  a covariant derivative compatible with the $G_2$ structure, this is not case.
On a manifold with a $G_2$ structure, there  is a two parameter family of covariant  derivatives compatible with a given $G_2$ structure on $Y$ \cite{Bryant:2005mz,delaOssa:2016ivz} with 
connection symbols
\[
\Gamma_{ab}{}^c = \Gamma^{\scriptscriptstyle{LC}}_{~ab}{}^{~c}  + A_{ab}{}^c(\alpha,\beta)~,\]
where 
$\Gamma^{\scriptscriptstyle{LC}}$ are the connection symbols of the Levi-Civita covariant derivative, $A_{abc}(\alpha, \beta)$ is the contorsion and $\alpha$ and $\beta$ are real parameters.   The contorsion is given by
\begin{align*} 
 A_{abc}(\alpha, \beta) &= \frac{1}{2}\, H_{abc} 
 - \frac{1}{6}\, \tau_{2\, da}\, \varphi_{bc}{}^d
 + \frac{1}{6}\, (1 + 2\beta)\, ( 
(\tau_1\lrcorner\psi)_{abc} - 4\, \tau_{1\, [b}\, g_{\varphi\,c]a})
\nn\\[5pt]
&\quad + \frac{1}{4}\, (1+2\alpha) \, 
(3\, \tau_{3\, abc} -2\, S_a{}^d\, \varphi_{bcd})~,
\end{align*}
where  $S$ is the traceless symmetric matrix corresponding to the torsion class $\tau_3$
\[ \tau_3 = \frac{1}{2}\, S^a\wedge\varphi_{abc}\, \dd x^{bc}\in \Lambda^3_{27}~.\] 
It is straightforward to show  that in fact, only the first two terms of the contorsion contribute to 
the right hand side of equations \eqref{eq:Liephi} and \eqref{eq:Liepsi}.  In other words, we only need to  work with a covariant derivative $\nabla$ with
\[ A_{abc} = \frac{1}{2}\, H_{abc} - \frac{1}{6}\, \tau_{2\, da}\, \varphi_{bc}{}^d~,\]
that is, with a connection with torsion
\[ T_{abc} = H_{abc} + \frac{1}{6}\, \tau_{2\, dc}\, \varphi_{ab}{}^d~.\]
The torsion is totally antisymmetric when $\tau_2 = 0$ and this corresponds to the unique covariant derivative with totally antisymmetric torsion. In this paper we are concerned mainly with integrable $G_2$ structures and hence we work with a connection for which $T = H$.


\section{Infinitesimal deformations of manifolds with an integrable $G_2$ structure}
\label{sec:finitemod}

We now turn to studying the tangent space to the moduli space of manifolds with an integrable $G_2$ structure. 
Finite deformations will be discussed in a future publication \cite{delaOssa17}.   In this section we discuss the infinitesimal deformations 
in terms of one forms $M_t$ with values in $TY$ and find moduli equations in terms of these forms.  Our main result is  that such deformations preserve the integrable $G_2$ structure if and only if $M_t$ satisfies equation \eqref{eq:integmod}. In addition, we derive equations for the variation of the intrinsic torsion of the manifold.

\subsection{Equations for deformations that preserve an integrable $G_2$ structure}
\label{sec:infinite}

Let $Y$ be a manifold with an integrable $G_2$ structure determined by $\varphi$. 
In this subsection we find equations that are satisfied by those infinitesimal deformations of the integrable $G_2$ structure which preserve the integrability.

From the discussion in section \ref{sssec:formdecomp} we can deduce that the infinitesimal deformations of the integrable $G_2$ structure take the form
\begin{align}
\partial_t\varphi &= \frac{1}{2}\, M^a_t\wedge\varphi_{abc} \, \dd x^{bc}= i_{M_t}(\varphi) 
~,\label{eq:varphi2}
\\[5pt]
 \partial_t\psi &= \frac{1}{3!}\, N^a_t\wedge\psi_{abcd}\, \dd x^{bcd} = i_{N_t}(\psi)~.
 \label{eq:varpsi2}
 \end{align}
 where $N_t$ and $M_t$ are one forms valued in $TY$. 
 The forms $N_t$ and $M_t$ are not independent as $\psi$ and $\varphi$ are Hodge dual to each other. To first order, $N_t$ and $M_t$ must be related such that 
 \[ \partial_t\psi = \partial_t*\varphi~.\]
We proved in \cite{delaOssa:2016ivz} that 
 the first order variations of the metric in terms of $M_t$ are given by
 \begin{align}
 \partial_t g_{\varphi\, ab} &= 2\, M_{t\, (ab)}
 ~,\label{eq:metricvar}
 \\[3pt]
  \partial_t \sqrt{\det g_{\varphi}} &= (\tr M_t) \, \sqrt{\det g_{\varphi}}~,
  \label{eq:detgvar}
  \end{align} 
 and that
 \[ M_t = N_t~.\]
 Note that only the symmetric part of $M_t$ contributes to the infinitesimal deformations of the metric.  To first order, we can interpret the antisymmetric part of $M_t$ as deformations of the $G_2$ structure  which leave the metric fixed, however this is not true at higher orders in the deformations as will be discussed in \cite{delaOssa17}.
 We give the equations for moduli of integrable $G_2$ structures in the following proposition.

\begin{proposition}
Let $Y$ be a manifold with an integrable $G_2$ structure $\varphi$ and $\psi = *\varphi$.
The infinitesimal moduli $M_t\in\Lambda^1(Y,TY)$ which preserve the integrability of the $G_2$ structure satisfy the equations
\begin{align}
 i_{\sigma_t}(\varphi) =  0
~,\label{eq:intphi}\\[5pt]
 i_{\sigma_t}(\psi)  = 0~,\label{eq:intpsi}
\end{align}
where $\sigma_t\in\Lambda^2(Y,TY)$ is given by
\be
\sigma_t = \dth M_t - [\hat H , M_t] - \partial_t \hat H~,
\label{eq:sigmat0}
\ee
or equivalently
\be \label{eq:sigmat}
\sigma_t^a = (\nabla_b\, M_{t\, c}^{\, a})\,\dd x^{bc} -
\partial_t \hat H^a ~,
\ee
where $\dth$ is an exterior covariant derivative on $TY$ with connection one form
\[ \theta_a{}^b = \Gamma_{ac}{}^b\, \dd x^c~,\]
and $\, \Gamma$ are the connection symbols of a connection $\nabla$ on $Y$ which is compatible with the $G_2$ structure and has totally antisymmetric torsion $H$
given by equation \eqref{eq:H}.
\end{proposition}

\begin{proof}
The proof of this proposition follows from the variations of equations \eqref{eq:dphi} and \eqref{eq:dpsi}. 

Consider first equation \eqref{eq:dphi}.  We can write the variation of the left hand side as
\[
\dd\partial_t \varphi = \dd\, i_{M_t}(\varphi)~.
\]
By equation \eqref{eq:Liephi} we find
\begin{equation}
\dd\partial_t \varphi  
= [\dd, i_{M_t}](\varphi) + i_{M_t}\dd\varphi
=   i_{\dth M_t}(\varphi) + i_{M_t}(i_{\hat H}(\varphi))
~,\label{eq:LHSmod} 
\end{equation}
where $\dth$ is an exterior covariant derivative on $TY$ with connection one form
\[ \theta_a{}^b = \Gamma_{ac}{}^b\, \dd x^c~,\]
and $\, \Gamma$ are the connection symbols of a connection $\nabla$ on $Y$ which is compatible with the $G_2$ 
structure and has totally antisymmetric torsion $H$ (see subsection \ref{ssec:applyG2}).
Now varying the right hand side, we have 
\[
\partial_t (i_{\hat H}(\varphi)) =  i_{\partial_t \hat H}(\varphi) + i_{\hat H}(i_{M_t}(\varphi))~.
\]
Equating this with \eqref{eq:LHSmod} we obtain
\[ 
 i_{\dth M_t - \partial_t \hat H}(\varphi) + [i_{M_t}, i_{\hat H}](\varphi)
= 0~.
\]
Equation \eqref{eq:intphi} follows this together with equation \eqref{eq:insbracket}
\[
 i_{\dth M_t - \partial_t \hat H - [\hat H, M_t]}(\varphi)
=0~.
\] 
where $[\hat H, M_t]$ is the Nijenhuis-Richardson bracket of $\hat H$ and $M_t$ as defined in equation \eqref{eq:NijRich}.  
Similarly one can obtain equation \eqref{eq:intpsi} by varying starting equation \eqref{eq:dpsi}.

To obtain \eqref{eq:sigmat} we need to write the exterior  derivative $\dth$ in terms of the covariant derivative.  Using \eqref{eq:NijRich} we have
\begin{align*}
\dth M_t^a - [\hat H, M_t]^a &= \dd M_t^a + \theta_b{}^a\wedge M_t^b
- \hat H^e\, M_{t\, e}^{~a} + M_{t\, b}^{~e}\,H_{ec}{}^a\, \dd x^{bc}
\\
&= \left(\partial_b M_{t\, c}^{~a} + \Gamma_{eb}{}^a\, M_{t\, c}^{~e}
- \frac{1}{2}\, H_{bc}{}^e\, M_{t\, e}^{~ a}
+ H_{be}{}^a\, M_{t\, c}^{~e}
\right)\, \dd x^{bc}
\\[5pt]
&= \left(\nabla^{\scriptscriptstyle{LC}}_b M_{t\, c}^{~a} 
+ \frac{1}{2}\, H_{be}{}^a\, M_{t\, c}^{~e}
- \frac{1}{2}\, H_{bc}{}^e\, M_{t\, e}^{~ a}
\right)\, \dd x^{bc}
= \nabla_b M_{t\, c}^{~a}\, \dd x^{bc}
\end{align*}

\end{proof}

We have shown that forms $M_t\in\Lambda^1(Y,TY)$ satisfying equations \eqref{eq:intphi} and \eqref{eq:intpsi} are infinitesimal moduli of manifolds with an integrable $G_2$ structure.   Even though this paper is concerned with heterotic compactifications,  the moduli problem  described in this section will have applications in other contexts in mathematics and in string theory.  In order to understand better the content of these equation we make here a few  remarks. Consider first equation \eqref{eq:intpsi} which, as a five form equation, can be decomposed into irreducible  representations of $G_2$.  
Using identities \eqref{eq:5insert7} and \eqref{eq:5insert14}, one can prove that this equation becomes \cite{delaOssa:2016ivz}
\begin{align}
\pi_7(i_\sigma(\psi)) &=  - \pi_{14}(\sigma^a)_{ab} \wedge\psi \nn\\
& = \big(4\, (\partial_t\tau_1 + i_{M_t}(\tau_1)) + (\pi_{14}(\dth M_t^a))_{ba}\dd x^b\big)\wedge\psi = 0~,
\label{eq:tau1red}\\
\pi_{14}(i_\sigma(\psi)) &= i_{\pi_7(\sigma)}(\psi) = i_{\cdth M_t}(\psi) = 0~.\label{eq:integmod}
\end{align}
The second equation represents deformations of the integrable $G_2$ structure which preserve the integrability and it is in fact the {\it only} constraint on $M_t$.  Observe how $\pi_7([\hat H, M_i] + \partial_t \hat H)$ drops out from this equation automatically
\[
i_{\pi_7([\hat H, M_i] + \partial_t \hat H)}(\psi) = 0~.
\]
 The first equation \eqref{eq:tau1red} then gives the variation of $\tau_1$ for given a solution of \eqref{eq:integmod}.  The other equation for moduli, equation \eqref{eq:intphi} gives the variations of all torsion classes for each solution of equation \eqref{eq:integmod}. Consequently, it does not restrict $M_t$. 
We note that equation \eqref{eq:tau1red} is in fact redundant as its contained in \eqref{eq:intphi}. 
It is important to remark too that, as equation \eqref{eq:integmod} is the only constraint on the variations of the integrable $G_2$ structure,  
there is no reason to expect that this space is finite dimensional (except of course in the case where $Y$ has $G_2$ holonomy). 
 
The tangent space to the moduli space of an integrable $G_2$ structure is found by modding out the set of solutions to equation \eqref{eq:integmod} by those which correspond to trivial deformations, that is diffeomorphisms.   These trivial infinitesimal deformations of $\varphi$ and $\psi$ are given by the 
Lie derivatives of $\varphi$ and $\psi$ respectively along a vector field $V$.  By equations \eqref{eq:Liephi} and \eqref{eq:Liepsi} these are given by
\begin{align}
{\cal L}_V(\varphi) &= [\dd, i_V](\varphi) = i_{\dth V}(\varphi)~,\\
{\cal L}_V(\psi) &= [\dd, i_V](\psi) = i_{\dth V}(\psi)~.
\end{align}
Therefore trivial deformations $M_{triv}$ of the $G_2$ structure correspond to
\be
  M_{triv} = \dth V~.
  \label{eq:trivialg2}
  \ee
The decompositions of ${\cal L}_V(\varphi)$  and ${\cal L}_V(\psi)$ into irreducible representations of $G_2$ are given by (see equations \eqref{eq:three1}-\eqref{eq:three27})
\begin{align}
\tr M_{triv} &= \nabla^{\scriptscriptstyle  {LC}}_a v^b = - \dd^\dagger v~,\label{eq:trivtr}
\\
M_{triv\, (ab)} &=  \nabla^{\scriptscriptstyle  {LC}}_{(a} v_{b)}~,\label{eq:trivsym}
\\
\pi_7 (m_{triv}) &= - \frac{1}{2}\,\pi_7\left( \dd v +  v\lrcorner H\right)~.\label{eq:trivantisym}
\end{align}

Therefore, the tangent space to the moduli space of deformations of integrable $G_2$ structures is given by the solutions of equation \eqref{eq:integmod} modulo the trivial variations of the $G_2$ structure given by equation \eqref{eq:trivialg2}.  We will call this space ${\cal TM}_0$.  As mentioned earlier, there is no reason why the resulting space of infinitesimal deformations is finite dimensional, unless one restricts to special cases such as $Y$ having $G_2$ holonomy. 

Finally, we would like to note on a property of the curvature of a manifold with an integrable $G_2$ structure. For any trivial deformation $M_{triv} = \dth V$, equation \eqref{eq:integmod} gives
\[
i_{\cdth^2 V}(\psi) = 0~.
\]
Therefore,
\be
i_{\check R(\theta)}(\psi) = 0~,\label{eq:thetacurv}
\ee
where $R(\theta)$ is the curvature of the one form connection $\theta$ and $\check R(\theta) = \pi_7(R(\theta))$. This equation is not an extra constraint, but in fact \eqref{eq:thetacurv} turns out always to be true when the $G_2$ structure is integrable. Indeed, covariant derivatives of the torsion classes are related to the curvature
two form, and can be used to show \eqref{eq:thetacurv} without any discussion of the deformation problem. We include the computation in appendix \ref{sec:appcurvs}, leading to \eqref{eq:curvintG2}.

\subsection{A reformulation of the equations for deformations of $G_2$ structures}\label{sec:42}

In section \ref{sec:InfDefHet}, we will determine the moduli space of heterotic $G_2$ systems. To this end, it is useful  to solve for $\sigma_t\in\Lambda^2(Y,TY)$ in equations \eqref{eq:intphi} and \eqref{eq:intpsi}.  We have the following lemma

\begin{lemma}
Let $\sigma\in\Lambda^2(Y, TY)$ and define
\begin{equation*}
\lambda = i_{\sigma}(\varphi)
~\qquad
\Lambda = i_{\sigma}(\psi)~.
\end{equation*}
Then $\sigma$ satisfies  $\Lambda=0$ and $\lambda= 0$, if and only if
\begin{equation}
  (\check\sigma_{a}\lrcorner\varphi)_{b}
   = (\sigma^d)_{ca}\, \varphi_{bd}{}^c~,
  \label{eq:idsigma}
   \end{equation}
  where $\check\sigma = \pi_7\sigma$.
\end{lemma}

\begin{proof}
The Hodge dual of $\Lambda$ can be easily computed (using equation \eqref{eq:g2psiep}) and is given by
\[ *\Lambda =  - \frac{1}{2}\, \big(
(\sigma^c)_{cd}\, \varphi^d{}_{ab} + 2\, (\sigma_a\lrcorner\varphi)_b
\big)\, \dd x^{ab}
\]
Therefore $\Lambda = 0$ is equivalent to
\begin{equation*}
 (\sigma_{[a}\lrcorner\varphi)_{b]} 
 = - \frac{1}{2}\,(\sigma^c)_{cd}\, \varphi_{ab}{}^d
 ~.
\end{equation*}
Note that contracting this equation with $\varphi^{ab}{}_e$ we find that
\begin{equation}
 (\pi_{14}(\sigma^a))_{ab} = 0~,\label{eq:pi14traces}
 \end{equation}
and so
\begin{equation}
 (\check\sigma_{[a}\lrcorner\varphi)_{b]} 
 = - \frac{1}{2}\,(\check\sigma^c)_{cd}\, \varphi_{ab}{}^d
 ~.\label{eq:premod1}
\end{equation}
where $\check\sigma = \pi_7(\sigma)$.
We now decompose the four form $\lambda$ into representations of $G_2$ as in subsection \ref{sssec:formdecomp}, and set each component to zero. The components of $\lambda$ are obtained by the following computation (see equations \eqref{eq:fourtoM}-\eqref{eq:four27}
\[
\frac{1}{12}\, \psi^{cde}{}_a\, \lambda_{bcde}
=  \frac{1}{2}\, \psi^{cde}{}_a\, (\sigma^f)_{[bc}\, \varphi_{de]f}
= \frac{1}{4}\, \psi^{cde}{}_a\, \big(
(\sigma^f)_{bc}\, \varphi_{def} -  (\sigma^f)_{cd}\, \varphi_{ebf}\big)~.
\]
Using the identity \eqref{eq:g2ident5} in the second term
\begin{align*}
\frac{1}{12}\, \psi^{cde}{}_a\, \lambda_{bcde}
&= \frac{1}{4}\, \big(
4\, (\sigma^f)_{bc}\, \varphi^c{}_{af}
- 6\, (\sigma^f)_{cd}\, g_{ar}\, \delta^{[c}_{[b}\, \varphi^{dr]}{}_{f]}
\big)
\\[5pt]
&= \frac{1}{2}\, \big(
- 2\, (\sigma_c)_{db}\, \varphi^{cd}{}_{a}
- (\sigma^e)_{cd}\, (- 2\, \delta^{c}_{[b}\, \varphi^d{}_{e]a}
+ g_{a[b}\, \varphi^{cd}{}_{e]})
\big)
\\[5pt]
&= \frac{1}{2}\, \big(
-  (\sigma_c)_{db}\, \varphi^{cd}{}_a
+ (\sigma^c)_{cd}\, \varphi^d{}_{ab}
- 2\, g_{a[b}\, (\check\sigma^e\lrcorner\varphi)_{e]}
\big)
\\[5pt]
&= \frac{1}{2}\, \big(
-  (\sigma_c)_{db}\, \varphi^{cd}{}_a
+ (\sigma^c)_{cd}\, \varphi^d{}_{ab}
- g_{ab}\, (\check\sigma^e\lrcorner\varphi)_e
+  (\check\sigma_a\lrcorner\varphi)_b
\big)
\end{align*}
Hence, $\lambda = 0$ is equivalent to
\[
0 = -  (\sigma_c)_{db}\, \varphi^{cd}{}_a
+ (\sigma^c)_{cd}\, \varphi^d{}_{ab}
- g_{ab}\, (\check\sigma^e\lrcorner\varphi)_e
+  (\check\sigma_a\lrcorner\varphi)_b~.
\]
Taking the trace of this equation gives
\[  (\check\sigma_a\lrcorner\varphi)^a = 0~,\]
and therefore 
\begin{equation}
0 = -  (\sigma_c)_{db}\, \varphi^{cd}{}_a
+ (\check\sigma^c)_{cd}\, \varphi^d{}_{ab}
+  (\check\sigma_a\lrcorner\varphi)_b~,
\label{eq:premod2}
\end{equation}
where we have used equation \eqref{eq:pi14traces} in the second term.

So far, we have proved that $\lambda = 0$ and $\Lambda = 0$ are equivalent to 
equations \eqref{eq:premod1}  \eqref{eq:premod2}.
Taking the antisymmetric part of equation \eqref{eq:premod2} we have
\begin{equation*}
0 = -  (\sigma_c)_{d[b}\, \varphi^{cd}{}_{a]}
+ (\check\sigma^c)_{cd}\, \varphi^d{}_{ab}
+  (\check\sigma_{[a}\lrcorner\varphi)_{b]}~,
\end{equation*}
and using \eqref{eq:premod1} in the third term we find
\begin{equation*}
(\check\sigma^c)_{cd}\, \varphi^d{}_{ab}
=
 2\, (\sigma_c)_{d[b}\, \varphi^{cd}{}_{a]}~.
\end{equation*}
Using this back into equation \eqref{eq:premod2} we have
\[
0 = -  (\sigma_c)_{db}\, \varphi^{cd}{}_a
+ 2\, (\sigma_c)_{d[b}\, \varphi^{cd}{}_{a]}
+  (\check\sigma_a\lrcorner\varphi)_b~,
\]
from which \eqref{eq:idsigma} follows. 

\end{proof}

The result of the lemma is that $\sigma_t$ defined as in \eqref{eq:sigmat} 
satisfies
\[
((\check\sigma_{t\, a})_{cd} - 2\, (\sigma_{t\, c})_{da})\, \varphi^{cd}{}_b = 0~.
\]
In other words, defining a two form $\Sigma_t\in\Lambda^2(Y,TY)$ by
\[
\Sigma_{t\, a} = 
\frac{1}{4}\, \big( (\sigma_{t\, a})_{bc} 
- 2\, (\sigma_{t\, b})_{ca} \big)\, \dd x^{bc}
= \frac{1}{2}\,\left( \sigma_{t\, a} - (\sigma_{t\, b})_{ca}\, \dd x^{bc}
\right)
~,\]
the equation for moduli is equivalent to 
\[ \check\Sigma_t = \pi_7(\Sigma_t) = 0~.\]
We would like to write this equation in terms of $M_t$ and $H$. We have
\begin{align*}
\Sigma_{t\, a} &= \frac{1}{2}\, \left(\sigma_{t\, a} - (\sigma_{t\, b})_{ca}\, \dd x^{bc}\right)
= \sigma_{t\, a} -
\frac{1}{4}\, \left(
2\, (\sigma_{t\, b})_{ca}
+ (\sigma_{t\, a})_{bc}
\right)\, \dd x^{bc}
\\[5pt]
&= \dth M_{t\, a} - [\hat H, M_t]_a - g_{ae}\,  (\partial_t\hat H^e)
- \frac{3}{4}\, (\sigma_{t\, [a})_{bc]} \, \dd x^{bc}~.
\end{align*}
The last two terms of this equation become after using equation \eqref{eq:sigmat} in the last term, 
\begin{align*}
-\frac{3}{4}\, (\sigma_{t\, [a})_{bc]} \, \dd x^{bc} - g_{ae}\,  (\partial_t\hat H^e)
&= -\frac{3}{2}\, \left( 
\nabla_{[b}M_{t\, |a|c]} - \frac{1}{2}\, g_{e[a}\, \partial_t H_{bc]}{}^e\right) \, \dd x^{bc} -  g_{ae}\,  (\partial_t\hat H^e)
\\[5pt]
&= -\frac{3}{2}\, \left( 
\partial_{[b}m_{t\, ac]} + H_{[ba}{}^e\, m_{t\, c]e}
\right) \, \dd x^{bc}
\\[7pt]
&~\quad + \frac{3}{4}\, \left(
\partial_t H_{abc} - 2\, M_{t\, (e[a)}\, H_{bc]}{}^e
\right) \, \dd x^{bc}
- \partial_t\hat H_a + 2\, M_{(ae)}\, \hat H^e
\\[5pt]
&= 
-\frac{3}{2}\, \left( 
- H_{[ab}{}^e\, m_{t\, c]e}
+ H_{[bc}{}^e\, M_{t\, (a]e)}
\right) \, \dd x^{bc} 
+ 2\,  M_{(ae)}\, \hat H^e
\\[5pt]
&\qquad +  (\dd m_t)_a + \frac{1}{2}\, \partial_t \hat H_a
\\[5pt]
&= 
[\hat H, M_t]_a 
+  (\dd m_t)_a+ \frac{1}{2}\, \partial_t \hat H_a~.
\end{align*}
Therefore 
\[ \Sigma_{t\, a} = \dth M_{t\, a} + (\dd m_t)_a 
 + \frac{1}{2}\, \partial_t \hat H_a~,\]
 where $m_t$ is the two form obtained from the antisymmetric part of $M_t$,
 that is,
 \[ 
 m_t = \frac{1}{2}\, M_{t\, [ab]}\, \dd x^{ab}~,
 \]
as in equation \eqref{eq:handm} in section \ref{sssec:formdecomp}.
 The equation for moduli for a manifold $Y$ with an integrable  $G_2$ structure is
\begin{equation}
0 = \Sigma_{t\, a}\lrcorner\varphi 
= \left(\dth\, M_{t\, a}
+ (\dd m_t)_a 
 + \frac{1}{2}\, \partial_t \hat H_a
\right)\lrcorner\varphi
~.
\label{eq:modulieq}
\end{equation}
This equation cannot depend on $\pi_{14}(m)$ as these are not part of the moduli of the integrable $G_2$ structure as discussed before (see subsection \ref{sssec:formdecomp}).  To check that in fact $\pi_{14}(m)$ drops off equation \eqref{eq:modulieq}, we prove the following lemma.

 \begin{lemma}\label{lemma:lemz}
Let $z$ be a one form with values in $T^*Y$ such that the matrix $z_{ab} = (z_a)_b$ is antisymmetric.
Then
\begin{equation*}
\dth z_a = - (\dd z)_a + \frac{1}{2}\, (\nabla_a z_{bc} )\, \dd x^{bc}~,
\end{equation*}
 where
\[ 
z = \frac{1}{2}\, z_{ab}\, \dd x^{ab}~.\]
If moreover $z\in \Lambda_{14}^2$,
we have
\begin{equation*}
(\dth z_a + (\dd z)_a)\lrcorner\varphi = 0
\end{equation*}

\end{lemma}
\begin{proof}
For the first identity we have
\begin{align*}
\dth z_a &= (\partial_b z_{ac} - \Gamma_{ab}{}^e\, z_{ec})\, \dd x^{bc}
= \frac{1}{2}\, (3\, \partial_{[b} z_{ac]} + \partial_a z_{bc}- 2\, \Gamma_{ab}{}^e\, z_{ec})\, \dd x^{bc}
\\[5pt]
&= - (\dd z)_a  +  \frac{1}{2}\, (  \nabla_a z_{bc} )\, \dd x^{bc}~.
\end{align*}
The second identity follows from the fact that if $z\in\Lambda_{14}^2$, then $z\lrcorner\varphi = 0$. 

\end{proof}

Note in particular that when we restrict to the $G_2$ holonomy case with vanishing flux ($H=0$), the moduli equation \eqref{eq:modulieq} reduces to
\begin{equation}
0 = \Sigma_{t\, a}\lrcorner\varphi 
= \left(\dth\, M_{t\, a}
+ (\dd m_t)_a
\right)\lrcorner\varphi
~.
\label{eq:modulieqG2Hol}
\end{equation}
where now $\dth$ denotes the Levi-Civita connection. As shown in \cite{delaOssa:2016ivz}, one can always make a diffeomorphism gauge choice where
\begin{equation}
\check\dth\, h_{t\, a}=\check\dth^\dagger\, h_{t\, a}=0\:,\;\;\Leftrightarrow\;\;h_{t\, a}\in{\cal H}_{\check\dth}^{1}(TY)\cong H_{\check\dth}^1(TY)\:,
\end{equation}
where $h_t$ is the symmetric traceless part of $M_t$, and ${\cal H}_{\dth}^*(TY)$ denote $\dth$-harmonic forms. Note that $h_t$ is restricted to the $\bf 27$ representation of ${\cal H}_{\check\dth}^{1}(TY)$. The remaining representations are the singlet $\bf 1$, which corresponds to trivial re-scalings of the metric, and the anti-symmetric $\bf 14$ representation, which in string theory have a natural interpretation as $B$-field deformations. 

For completeness, but not relevant to the work in this paper, we note that the procedure in this section can also be used to find infinitesimal deformations of a manifold $Y$ with a $G_2$ structure which is not necessarily integrable.  The result in this case is 
\begin{align*}
0 &= \Sigma_{t\, a}\lrcorner\varphi
\\
&=\left(\dth\, M_{t\, a}
+ (\dd m_t)_a 
 + \frac{1}{2}\, \partial_t \hat H_a
\right)\lrcorner\varphi 
+  \frac{1}{2}\, (\partial_t \tau_{2\, ab} + M_{t\, b}^{\, e}\, \tau_{2\, ea})\, \dd x^b~.
\end{align*}
In this case, all these equations give the deformations of the torsion classes in terms of $M_t$. Infinitesimal deformations of a $G_2$ structure give another $G_2$ structure as 
the existence of a $G_2$ structure on $Y$ is a topological condition (in fact, any 7-dimensional manifold which is spin and orientable, that is, its first and second Stiefel-Whitney classes are trivial, admits a G2 structure).  

A couple of remarks are in order regarding the equations for moduli obtained in this section.  What we have demonstrated is that equation \eqref{eq:modulieq} is equivalent to equations \eqref{eq:intphi} and \eqref{eq:intpsi}.  On a first sight, equation \eqref{eq:modulieq} looks useless as we do not have (at this stage) an independent way to describe the variations of the torsion in terms of the $M_t$.  Equation \eqref{eq:modulieq} however will become useful in section \ref{sec:InfDefHet} when we discuss the moduli of heterotic $G_2$ systems. In this context,  perturbative quantum corrections to the theory require the cancelation of an anomaly which gives an independent description of  $H$ in terms of instanton connections on both $TY$ and a vector bundle $V$ on $Y$ .


\section{Moduli space of instantons on manifolds with $G_2$ structure}\label{sec:instantons}

We now turn to studying the moduli space of integrable $G_2$ manifolds with instantons. There is a large literature on deformations of instantons on manifolds with special structure \cite{reyes1993some, carrion1998generalization, donaldson1998gauge, 2000math.....10015T,2009arXiv0902.3239D, sa2009instantons, 2014SIGMA..10..083S,Ball:2016xte,Harland:2009yu,2010JHEP...10..044B,Harland:2011zs,Ivanova:2012vz,Bunk:2014ioa,Bunk:2014kva,Haupt:2014ufa,2015arXiv151104928H,Charbonneau:2015coa,delaOssa:2016ivz, Munoz:2016rgc}. In order for this paper to be self-contained, we will now review the results of \cite{delaOssa:2016ivz}, using the insertion operators  introduced in previous sections. We will see that, in this set up,  proofs of the theorems of \cite{delaOssa:2016ivz} simplify drastically. 

Consider a one parameter family of pairs $(Y_t, V_t)$ with $(Y_0, V_0)=(Y, V)$, 
$V$ is vector bundle over a manifold $Y$ which admits an integrable $G_2$ structure.
Let $F$ be the curvature of $V$ and we take $F$ to satisfy the instanton equation
\be
F\wedge\psi = 0 ~.\label{eq:instanton}
\ee 
The moduli problem that we want to discuss in this section is the  simultaneous deformations of the integrable $G_2$ structure on $Y$ together with those of the bundle $V$ which preserve both the integrable $G_2$ structure on $Y$ and the instanton equation.   We begin by considering variations of equation \eqref{eq:instanton}.

\begin{theorem}[\cite{delaOssa:2016ivz}]
Let $M_t\in \Lambda^1(TY)$ be a deformation of the integrable $G_2$ structure on $Y$ and $\partial_t A$ a deformation of the instanton connection on $V$.  The simultaneous deformations $M_t$ and $\partial_t A$ which respectively preserve the integrable $G_2$ structure and the instanton condition on $V$ must satisfy
\be
\left(\dd_A\partial_t A - i_{M_t}(F)\right)\lrcorner\varphi = 0~.
\label{eq:Atiyah}
\ee
\end{theorem}
\begin{proof}
Variations of the instanton equation \eqref{eq:instanton} give
\[ 0 = \partial_t(F\wedge\psi) = \partial_t F\wedge\psi + F \wedge\partial_t\psi~.\]
Note that in the first term, the wedge product of $\partial_t F$ with $\psi$ picks out the part of  $\partial_t F$ which is in $\Lambda^2_7$. 
Noting that 
\[ \partial_t F = \dd_A\partial_t A~,\]
we obtain
\[
\dd_A\partial_t A \wedge\psi + F\wedge i_{M_t}(\psi) = 0~
\]
Taking the Hodge dual we obtain equivalently
\[
(\dd_A\partial_t A)\lrcorner\varphi = - *( F\wedge i_{M_t}(\psi)) 
= - *( F\wedge M_t^a\wedge\psi_a)
= *(M_t^a\wedge F_a \wedge\psi)
= *( i_{M_t}(F)\wedge\psi)~.
\]
where we have used the identity \eqref{eq:idtwo14} in the second to last equality.
Therefore the result follows.

\end{proof}

Note that $\partial_t A$  is not well defined (it is not an element of $\Lambda^1(Y, {\rm End}(V))$), however equation \eqref{eq:Atiyah} is covariant.  Under a gauge transformation $\Phi$, $A$ transforms as
\[
A\quad\mapsto\quad{}^\Phi A = \Phi\, ( A  - \Phi^{-1}\, \dd\Phi) \Phi^{-1}~,
\]
and hence $\partial_t A$ transforms as
\[
\partial_t\, A\quad\mapsto\quad {}^\Phi (\partial_t \, A) = \Phi\, \big(\partial_t A  - \dd_A(\Phi^{-1}\partial_t\Phi)\big) \Phi^{-1}~.
\]
After a short computation, we find
\[
\dd_A\partial_t A\quad\mapsto\quad{}^\Phi (\dd_A\partial_t A) = \Phi \big(\dd_A\partial_t A - \dd_A^2(\Phi^{-1}\partial_t\Phi)\big)\Phi^{-1}~,
\]
and contracting with $\varphi$
\begin{align*}
(\dd_A\partial_t A)\lrcorner\varphi
\quad\mapsto\quad {}^\Phi (\dd_A\partial_t A)\lrcorner\varphi &= 
\Phi \left(\big(\dd_A\partial_t A - \dd_A^2(\Phi^{-1}\partial_t\Phi)\big)\lrcorner\varphi\right)\Phi^{-1}
\\
&= \Phi \big((\dd_A\partial_t A)\lrcorner\varphi\big)\Phi^{-1}~,
\end{align*}
where we have used the fact that $\check\dd_A^2 = 0$.  Hence equation \eqref{eq:Atiyah} is covariant\footnote{This had already been noticed by Atiyah in connection with his work on the moduli of holomorphic bundles on complex manifolds \cite{atiyah1957complex} and has been used in \cite{Candelas:2016usb}. Here we generalise this point to the case at hand of the moduli of instanton connections on manifolds with an integrable $G_2$ structure}. 
One can define a {\it covariant deformation} of $A$, $\alpha_t\in\Lambda^1(Y,{\rm End}(V))$, by introducing a connection one form $\Lambda$ {\it on the moduli space} of instanton bundles over $Y$\footnote{Here we generalise the work of \cite{Candelas:2016usb} where covariant variations of holomorphic connections were constructed.}. Because equation \eqref{eq:Atiyah} is already a covariant equation for the moduli, it should be the case that
\be 
\alpha_t = \partial_t A - \dd_A\Lambda_t~,\label{eq:covvarA}
\ee
that is, $\alpha_t$ and $\partial_t A$ can only differ by a term which is $\check\dd_A$-closed.  Note that
$\alpha_t$ is in fact covariant as long as the connection $\Lambda_t$ transforms under gauge transformations as
\[
\Lambda_t\quad\mapsto \quad^\Phi\, \Lambda_t = \Phi\,(\Lambda_t - \Phi^{-1}\partial_t\Phi)\,\Phi^{-1}~.
\]
In terms of elements $\alpha_t\in\Lambda^1(Y, {\rm End}(V))$, equation \eqref{eq:Atiyah} is
\be
\left(\dd_A\alpha_t - i_{M_t}(F)\right)\lrcorner\varphi = 0~.
\label{eq:Atiyah2}
\ee

It will convenient (and important) to understand better the moduli problem to define  the map \cite{delaOssa:2016ivz}\footnote{Note that 
${\cal F}$ and $M_t$ have changed signs compared to \cite{delaOssa:2016ivz}.}
\begin{align*}
{\cal F} :\quad  \Lambda^p(Y, TY)\qquad &\longrightarrow\qquad \Lambda^{p+1}(Y, {\rm End}(V))\\
 M\quad\qquad&~\mapsto \qquad 
{\cal F}(M) = (- 1)^p\, i_M(F)~.
 \end{align*}
 We also define the map 
 \begin{equation*}
\check{\cal F} :\quad  \Lambda_{\bf r}^p(Y, TY)\qquad \longrightarrow\qquad \Lambda_{\bf r'}^{p+1}(Y, {\rm End}(V))~, 
\end{equation*}
where $\Lambda_{\bf r}^p(Y, {\rm End}(V))\subseteq \Lambda^p(Y, {\rm End}(V))$, 
$\Lambda_{\bf r'}^{p+1}(Y, {\rm End}(V))\subseteq \Lambda^{p+1}(Y, {\rm End}(V))$,  and ${\bf r}$ and ${\bf r'}$ are appropriate irreducible $G_2$ representations as follows:
\begin{align*}
 \check{\cal F}(M) &= {\cal F}(M) =  \, i_M(F) ~, 
 &{\rm for}\quad M \in \Lambda^0(TY)~,
 \\
 \check{\cal F}(M) &= \pi_7({\cal F}(M)) =  - \pi_7(i_M(F))~, &{\rm for}\quad M \in \Lambda^1(TY)~,
 \\
  \check{\cal F}(M) &= \pi_1({\cal F}(M)) =  \pi_1(i_M(F))~, &{\rm for}\quad M \in \Lambda_7^2(TY)~.
\end{align*}
Note that the projections that define $\check{\cal F}$ are completely analogous to those that define the derivatives $\check \dd_A$.  In terms of this map,
equation \eqref{eq:Atiyah2} can be written as
\be
\check\dd_A\alpha_t + \check{\cal F}(M_t) = 0~.
\label{eq:AtiyahBis}
\ee

The theorem below proves that  as a consequence of the Bianchi identity $\dd_A F = 0$, 
$\check{\cal F}$ maps the moduli space of manifolds with an integrable $G_2$ structure into the $\check \dd_A$-cohomology discussed in section \ref{sec:cancom}.

\begin{theorem}[\cite{delaOssa:2016ivz}] \label{theorem:instanton}
Let $M\in\Lambda^p(Y, TY)$, where $p = 0, 1, 2$, and let $F$ be the curvature of a bundle $V$ with one form connection $A$ which satisfies the instanton equation.  Let $\nabla$ be a covariant derivative on $Y$ compatible with the integrable $G_2$ structure on $Y$ 
with torsion $H$, and $\dth$ be an exterior covariant derivative such that
\[ \theta_a{}^b = \Gamma_{ac}{}^b\, \dd x^c~,\]
where $\Gamma$ are the connection symbols of $\nabla$.
Then the Bianchi identity 
\[ \dd_A F = 0~,\]
implies
\begin{equation}
\check \dd_A(\check{\cal F}(M))+ \check{\cal  F}(\check \dth(M)) = 0~.\label{eq:idcohom}
 \end{equation}
 Forms $M \in \Lambda^p(Y, TY)$ which are $\check \dth$-exact
 are mapped into forms in $\Lambda^{p+1}(Y, {\rm End}(TY))$ which are $\check \dd_A$-exact.
 Furthermore, {\bf any} form $M\in \Lambda^1(Y, TY)$ which satisfies the moduli equation
\[i_{\cdth M}(\psi) = 0~,\]
is mapped into a $\check \dd_A$-closed form in $\Lambda^2(Y, {\rm End}(TY))$.
Therefore, $\check{\cal F}$ maps the infinitesimal moduli space ${\cal TM}_0$ of $Y$ into elements of the cohomology
$H^2_{\check \dd_A}(Y, {\rm End}(V))$.
\end{theorem}

\begin{proof}
Consider $\dd_A i_M(F)$. Then
\[
\dd_A i_M(F) = [\dd_A, i_M](F) + i_M\dd_A F~.
\]
The second term vanishes by the Bianchi identity.  Using equation \eqref{eq:CNL} we find
\[
\dd_A i_M(F) - i_{\dth M}(F) = (-1)^p\, M^a\wedge \nabla^A_a F~,
\]
where 
\[ 
\nabla^A_a F = \frac{1}{2}\, \nabla^A_a (F_{bc})\, \dd x^{bc}~.
\]
Contracting with $\varphi$ we find
\begin{align*}
\big(\dd_A i_M(F) - i_{\dth M}(F)\big)\lrcorner\varphi
&= \big(\dd_A i_M(F) - {\cal F}({\dth M})\big)\lrcorner\varphi
= (-1)^p\, \big( M^a\wedge \nabla^A_a F\big)\lrcorner\varphi
\\
&= (-1)^p\ *\big( M^a\wedge \nabla^A_a F\wedge\psi\big)
\\
&=
 (-1)^p\ *\big( M^a\wedge ( \nabla^A_a (F\wedge\psi) - 
F\wedge \nabla_a\psi)\big) = 0~.
\end{align*}
Hence, by the definition of $\cal F$ we find
\be
\big(\dd_A({\cal F}(M)) - {\cal F}(\dth M)\big)\lrcorner\varphi = 0~.\label{eq:precohom}
\ee
which implies equation \eqref{eq:idcohom} upon considering the appropriate projections for each value of $p$.

Suppose  $M\in\Lambda^p(Y, TY)$ is $\cdth$-exact, that is 
\[M= \cdth V~,\]
 for some $V\in\Lambda^{p-1}(Y, TY)$.  We want to prove that $\check{\cal F}(\cdth V)$ is maped into 
 a $\check\dd_A$-exact form in $\Lambda^{p+1}(Y,{\rm End}(TY))$.  This is now obvious from equation \eqref{eq:idcohom}.
 
Consider now $M\in\Lambda^1(Y, TY)$ which satisfies the moduli equation \eqref{eq:intpsi}.  We want to prove that ${\cal F}(M)$
is $\check\dd_A$-closed.  According to  equation \eqref{eq:idcohom}, this means  we need to prove that 
\[ {\cal F}(\cdth M) = 0\]
when $M$ satisfies \eqref{eq:intpsi}.  This is in fact the case as can be verified by the following computation
\begin{align}
{\cal F}(\dth M)\lrcorner\varphi &= * ( {\cal F}(\dth M)\wedge\psi)
= *(i_{\dth M}(F)\wedge \psi) = * (\dth M^a\wedge F_a \wedge \psi)
\nn\\
& = - *( i_{\dth M}(\psi)\wedge F) 
= - *(i_{\cdth M}(\psi)\wedge F) = 0~.\label{eq:inproof}
\end{align}
In the second line of this computation we have used the identity  \eqref{eq:idtwo14} in the first equality,
and equations \eqref{eq:5insert7} and \eqref{eq:5insert14} in the second. 

\end{proof}

We remark that actually any $M \in\Lambda^1(Y, TY)$ which satisfies the moduli equation
\[
 i_{\sigma}(\psi)  = 0~,
 \]
 where
\[
\sigma_t = \dth M_t - [\hat H , M_t] - \partial_t \hat H \in\Lambda^2(Y,TY)~,
\] 
is mapped by $\check{\cal F}$ into a $\check\dd_A$-closed form.  Indeed, 
the last term in the calculation above in equation \eqref{eq:inproof} can be written as (see equations \eqref{eq:integmod} and \eqref{eq:tau1red})
\[
0 = - *(i_{\cdth M}(\psi)\wedge F) = - *(\pi_{14}(i_\sigma(\psi))\wedge F) = - *(i_\sigma(\psi)\wedge F)~.
\]

Equation \eqref{eq:AtiyahBis} and theorem \ref{theorem:instanton} give a very nice picture of the tangent space to the moduli space of simultaneous deformations of
the integrable $G_2$ structure on $Y$ together with the instanton condition on the bundle $V$ on $Y$. 
Keeping the $G_2$ structure fixed ($\partial_t\psi = 0$) on the base manifold equation \eqref{eq:AtiyahBis} gives
\begin{equation} 
\cd_A \alpha_t = 0~,\label{eq:Bmod}
\end{equation} 
which is the equation for the bundle moduli. It is also clear that variations of $A$ which are $\cd_A$-exact one-forms correspond to gauge transformations, so
the bundle moduli correspond to elements of the cohomology group 
\[ H_{\cd_A}^1(Y, \rm{End}(V))~.\]  

On the other hand, suppose that the parameter $t$ corresponds to a deformation of the integrable $G_2$ structure.  Then equation \eqref{eq:AtiyahBis} represents the equation that the moduli $M_t$ must satisfy in order for the instanton condition be preserved.  In fact,  it means that the variations $M_t\in {\cal TM}_0$ of the integrable $G_2$ structure of $Y$,
are such that $\check{\cal F}(M_t)$ must be 
$\check \dd_A$-exact, that is
\[ M_t\in {\rm ker}(\check{\cal F})\subseteq {\cal TM}_0~.\]
  Therefore, the tangent space of the moduli space of the combined deformations of the integrable $G_2$ structure and bundle deformations is given by
\be
 {\cal T}{\cal M}_1 = {\rm ker}(\check{\cal F})\oplus H^1_{\check \dd_A}(Y, {\rm End}(V))~,\label{eq:TM1}
 \ee
where elements in $H^1_{\check \dd_A}(Y, {\rm End}(V))$ correspond to bundle moduli.   Note again that there is no reason to believe that ${\rm ker}(\check{\cal F})$ is finite dimensional.

Finally, there is an important observation regarding the parts of  the moduli $M\in\Lambda^1(Y,TY)$ which appear in equation \eqref{eq:TM1}.  
Thinking about $M$ as a matrix, we have seen that $\pi_{14}(m)$ (where $m$ is the two form obtained from the antisymmetric part of $M$) drops out of the contractions $i_M(\psi)$ and $i_M(\varphi)$ corresponding to the variations of $\psi$ and $\varphi$ respectively.  Hence $\pi_{14}(m)$ plays no part in the moduli problem leading to ${\cal TM}_0$. It is easy to see that $\pi_{14}(m)$ also drops out from equations \eqref{eq:AtiyahBis} and \eqref{eq:idcohom}.  For equation \eqref{eq:AtiyahBis}, 
\[
{\cal F}(M)\lrcorner \varphi = - *(i_M(F)\wedge\psi)= - *(M^a\wedge F_a\wedge\psi) = *(M^a\wedge\psi_a\wedge F) = *(i_M(\psi)\wedge F)
\]
where we have used identity \eqref{eq:idtwo14}.  This same argument shows that $\pi_{14}(m)$ drops out of the first term of equation \eqref{eq:idcohom}.  As equation \eqref{eq:idcohom} must be true for any $M\in \Lambda^1(Y,TY)$, it follows that  $\pi_{14}(m)$ drops out of the second term too.

\section{Infinitesimal moduli of heterotic $G_2$ systems}
\label{sec:InfDefHet}
 
We now use the results of the previous sections to determine the infinitesimal moduli space of heterotic $G_2$ systems.  We show that the moduli problem can be reformulated in terms of a differential operator $\check{\cal D}$ acting on forms  ${\cal Z}$ with values in a bundle
\be
{\cal Q} = T^*Y\,  \oplus\, {\rm End}(TY)\oplus\, {\rm End}(V). 
\ee
We construct an exterior covariant derivative $\cal D$ by requiring that, for a one form ${\cal Z}$ with values in $\cal Q$,
the conditions $\check{\cal D}({\cal Z}) = 0$, reproduces the equations for moduli that we already have, that is equations \eqref{eq:modulieq} and \eqref{eq:AtiyahBis}. 
Furthermore, we show that $\check{\cal D}^2 = 0$ is enforced by the heterotic $G_2$ structure, including crucially equation \eqref{eq:idcohom}, and the anomaly cancelation condition that we introduce below. In other words, we show that the heterotic $G_2$ structure
corresponds to an instanton connection on $\cal Q$. Conversely, we prove that a differential which satisfies $\check{\cal D}^2 = 0$ 
implies the heterotic $G_2$ system including the (Bianchi identity of) the anomaly cancelation condition.  We show that this result is true to all orders in the $\alpha'$ expansion.
With this differential at hand, we show that the infinitesimal heterotic moduli space corresponds to classes in the cohomology group
\[ 
H^1_{\check {\cal D}}(Y, {\cal Q}) ~,
\]
which is finite dimensional.

 \subsection{The heterotic $G_2$ system in terms of a differential operator}
 
In this subsection we reformulate the heterotic $G_2$ system
\[ \left( [Y,\varphi],[V,A],[TY,\tilde{\theta}],H \right) \; ,
\]
in terms of a differential operator, or more precisely, a covariant operator $\check{\cal D}$, which acts on forms with values on the bundle  
\be
{\cal Q} = T^*Y\,  \oplus\, {\rm End}(TY)\oplus\, {\rm End}(V)~,\label{eq:Q}
\ee
and which satisfies $\check{\cal D}^2= 0$. It is important to keep in mind that we demand that $H$, which encodes the geometry 
of the integrable $G_2$ structure on $(Y,\varphi)$ (see equation \eqref{eq:H}) satisfies a constraint, the anomaly cancelation condition 
 \be
 H = \dd B + \frac{\alpha'}{4}\, ({\cal CS}[A] - {\cal CS}[\tilde\theta]) \; .
 \label{eq:Hdef}
\ee
In what follows we will also need the {\it Bianchi identity for the anomaly cancelation condition} which is obtained by applying the exterior derivative $\dd$ to the anomaly
\be
\dd H = \frac{\alpha'}{4}\, (\tr(F\wedge F) - \tr(\tilde R\wedge\tilde R))~
\label{eq:BIanom}
\ee
We show in appendix \ref{app:sugra} that heterotic $G_2$ systems correspond to certain vacua of heterotic supergravity, provided that  the torsion class $\tau_1$ is an exact form.  The results in this paper however apply to a more general system, as we do not assume anywhere that the torsion class $\tau_1$ is $\dd$-exact (by equation \eqref{eq:Intpsi} it is clear that for an integrable $G_2$ structure, $\tau_1$ is always $\cd$-closed). 

Consider the differential operator
 \begin{equation}
  {\cal D} =  \left(\,  
 \begin{matrix}
 \dth &  ~~{\cal \widetilde R} & ~  -{\cal F}~ 
 \\  {\cal \widetilde R} & ~~~\dd_{\tilde\theta} & ~~0
 \\ {\cal F} & ~~ 0  & ~~~\dd_A
 \end{matrix}
 \, \right)~,\label{eq:calD}
 \end{equation}
 which acts on forms with values in ${\cal Q}$. The operator acts linearly on forms with values in $\cal Q$ and it is easy to check that $\cal D$ satisfies the Leibniz rule, that is,
 \[
 {\cal D}(f\, {\cal V}) = \dd f\wedge{\cal V} + f\, {\cal D}{\cal V}~,
 \]
  for any section $\cal V$ of $\cal Q$ and any function $f$ on $Y$.
 Therefore, it defines a connection, or more appropriately, a covariant exterior derivative on $\cal Q$.  Its action on higher tensor products of ${\cal Q}$ can be obtained from the Leibniz rule. It is important to keep in mind in the definition of $\cal D$ that the two connections $\theta$ and $\tilde{\theta}$ on $TY$ are not the same (see more details in appendix \ref{app:sugra} for the reasons of this difference in the supergravity theory).   

 The map $\cal F$ has been defined already in section \ref{sec:instantons} by its action on forms with values in $TY$.  In defining $\cal D$, we extend the
 definition of the operator $\cal F$ to act on forms with values in $\cal Q$ as follows.
 Let $y\in \Lambda^p(Y, T^*Y)$, and $\alpha\in \Lambda^p(Y,{\rm End}(V)$.  Then 
 \begin{align*}
 {\cal F} :\quad \Lambda^p(Y, T^*Y)\oplus \Lambda^p(Y, {\rm End}(V))
 &\longrightarrow 
 \Lambda^{p+1}(Y, {\rm End}(V))\oplus \Lambda^{p+1}(Y, T^*Y)
 \\
 \left(\begin{matrix}
 y  \\ \alpha 
 \end{matrix}\right)\qquad\qquad\quad
 &\mapsto
 \quad\qquad \qquad  
\left(\begin{matrix}
 {\cal F}(y) \\  {\cal F}(\alpha) 
 \end{matrix}\right)\qquad
  \end{align*}
  where
  \begin{align*}
  {\cal F}(y) &= (-1)^{p}\, g^{ab}\, y_a\wedge F_{bc}\, \dd x^c = (-1)^{p}\, i_y(F)~,
  \\
  {\cal F}(\alpha)_a &= (-1)^{p}~ \frac{\alpha'}{4}\, \tr (\alpha\wedge F_{ab}\, \dd x^b)~.
  \end{align*}
  The map $\cal \tilde R$ is defined similarly, but acts on forms valued in $\Lambda^p(Y, T^*Y)\oplus \Lambda^p(Y, {\rm End}(TY)$.
 We also define the maps $\check{\cal F}$  and $\check{\cal \widetilde R}$ as in section \ref{sec:instantons} by an obvious generalisation.   

We now show that the projection $\check{\cal D}$ of the operator $\cal D$ satisfies $\check{\cal D}^2 = 0$
for heterotic $G_2$ systems.  The Bianchi identity of the anomaly cancelation condition enters crucially in the proof. 
   
 \begin{theorem}
\label{theorem:Dinstanton}
For a heterotic $G_2$ system  $( [Y,\varphi],[V,A],[TY,\tilde{\theta}],H )$,
the operator $\cal D$ satisfies $\check{\cal D}^2 = 0$.
\end{theorem}

\begin{proof}
Computing the square of equation \eqref{eq:calD} we have
\begin{equation}
  {\cal D}^2  
   =  
 \left(\,
 \begin{matrix}
 \dth^2 + {\cal \widetilde R}^2 - {\cal F}^2 
 & ~~~\dth{\cal \widetilde R} + {\cal \widetilde R}\dd_{\tilde\theta}~  &~~~-(\dth{\cal F} + {\cal F}\dd_A)~
\\ 
  {\cal \widetilde R}\dth + \dd_{\tilde\theta}{\cal \widetilde R}
 &~~{\cal \widetilde R}^2 + \dd_{\tilde\theta}^2
 &-  {\cal \widetilde R}{\cal F}
 \\
~{\cal F}\dth + \dd_A{\cal F}
&  {\cal F}{\cal \widetilde R}
& ~~- {\cal F}^2 + \dd_A^2~ 
 \end{matrix}
\, \right)~,\label{eq:nilD}
 \end{equation}
 We want to prove that  $\check{\cal D}^2 = 0$.
 
Consider first the condition corresponding to the $(31)$ entry of \eqref{eq:nilD}
\[
\check{\cal F}(\cdth y) + \check\dd_A\check{\cal F}(y) = 0~, \quad\forall~ y\in\Lambda^p(Y, T^*Y)~.
\]
This has already been proven (see equation \eqref{eq:idcohom} and its proof in theorem \ref{theorem:instanton}).
The condition for the $(21)$ entry
\[
{\check{\cal \widetilde R}}(\cdth  y) + \check\dd_{\tilde\theta}\,{\check{\cal \widetilde R}}(y) = 0
~, \quad\forall~ y\in\Lambda^p(Y, T^*Y)~,
\]
is similarly satisfied.

We already know that $\check\dd_A^2 = 0$,  and $\check\dd_{\tilde\theta}^2 = 0$ so  the conditions
for the entries $(22)$ and $(33)$ are respectively
\[
\check{\cal F}^2(\alpha) = 0~,\qquad {\check{\cal \widetilde R}}^2(\kappa) = 0~,
\] 
for any $\alpha\in\Lambda^1(Y,{\rm End}(V)$ and any $\kappa\in\Lambda^1(Y,{\rm End}(TY)$.
These equations are in fact is true.  For the first one
 \begin{equation}
 {\cal F}^2(\alpha) = 
- \frac{\alpha'}{4}\, g^{ac}\, \big(\tr(\alpha\wedge F_{ab}\,\dd x^b)\big)\wedge F_{cd}\, \dd x^d~.
\label{eq:calFsquared}
 \end{equation}
 By equation \eqref{eq:i14to14} we see inmediately that 
 \[\check{\cal F}^2(\alpha) = 0~.\]
 The proof that  ${\check{\cal \widetilde R}}^2(\kappa) = 0$ follows similarly.
 It also follows from equation \eqref{eq:i14to14}, that  the proof of the conditions corresponding to the entries $(23)$ and $(32)$ 
 \[
 {\check{\cal \widetilde R}}(\check{\cal F}(\alpha)) = 0~,\qquad
 \check{\cal F}({\check{\cal \widetilde R}}(\kappa)) = 0~.
 \]
 is completely analogous.
 
 Consider now the condition corresponding to the $(13)$ entry of \eqref{eq:nilD}.  For any $\alpha\in\Lambda^1(Y,{\rm End}(V)$, we have
  \begin{align*}
 \dth{\cal F}(\alpha)_a + {\cal F}(\dd_A\alpha)_a &=
 \dd {\cal F}(\alpha)_a - \theta_a{}^b\wedge {\cal F}(\alpha)_b 
 + (-1)^{p+1}\, \frac{\alpha'}{4}\, \tr\big((\dd_A\alpha)\wedge F_{ab}\, \dd x^b\big)
 \\
 &= (-1)^p\,\frac{\alpha'}{4}\,  \tr\Big(
 \dd_A(\alpha\wedge F_{ab}\, \dd x^b) 
-  \theta_a{}^b\wedge \alpha\wedge F_{bc}\, \dd x^c
 \\
 &\qquad \qquad\qquad
 - (\dd_A\alpha)\wedge F_{ab}\, \dd x^b
 \Big)
 \\
 &=
 \frac{\alpha'}{4}\, \tr\Big(
 \alpha\wedge (\dd_AF_{ab}\, \dd x^b
 - \theta_a{}^b \wedge F_{bc}\, \dd x^c)
 \Big)
 \\[5pt]
 &= \frac{\alpha'}{4}\, \tr\Big( 
 \alpha\wedge (
 - (\dd_A F)_a + \partial_{A\, a}F - \theta_a{}^b \wedge F_{bc}\, \dd x^c)
 \Big)
  \\[5pt]
 &= \frac{\alpha'}{4}\, \tr\Big(
 \alpha\wedge  \nabla_a^{A} F
 \Big)
~,
 \end{align*} 
 where the last two equalities follow equations \eqref{eq:covid}, \eqref{eq:cool}
 and the Bianchi identity $\dd_A F = 0$ 
 (see also lemma 3 in  \cite{delaOssa:2016ivz}).
 Contracting with $\varphi$ and using the fact that $F\lrcorner\varphi = 0$, we obtain
 \begin{align*}
 (\dth{\cal F}(\alpha)_a + {\cal F}(\dd_A\alpha)_a)\lrcorner\varphi =
   \frac{\alpha'}{4}\, \tr\Big(
 \alpha\lrcorner\big(
 (\nabla_a^{A} F )
 \lrcorner\varphi\big)
 \Big)
 = \frac{\alpha'}{4}\, \tr\Big(
 \alpha\lrcorner\big(
 \nabla_a^{A} (F \lrcorner\varphi) \big)
 \Big) = 0
 ~,
 \end{align*}
 as required.  Clearly the proof for the $(12)$ entry is similar, so
 \[
 \cdth{\check{\cal \widetilde R}}(\kappa) + {\check{\cal \widetilde R}}(\dd_{\tilde\theta}\kappa) = 0~.
 \]
 
 Finally, for the entry $(11)$ we need to prove that, for any $y\in\Lambda^p(Y, T^*Y)$,
 \begin{equation}
 \cdth^2\, y - \check{\cal F}^2(y) + {\check{\cal \widetilde R}}{}^2(y)= 0~,\label{eq:nil11}
 \end{equation}
 We have
 \begin{align*}
 \dth^2\, y_a &= - R(\theta)_a{}^b\wedge y_b~,
 \\
 {\cal F}^2(y)_a &= - \frac{\alpha'}{4}\, y^c\wedge \tr(F_{ab}\, \dd x^b\wedge F_{cd}\, \dd x^d)~,
 \\[5pt]
 {\cal \widetilde R}{}^2(y)_a 
 &= - \frac{\alpha'}{4}\, y^c\wedge \tr(\tilde R_{ab}\, \dd x^b\wedge \tilde R_{cd}\, \dd x^d)~,
 \end{align*}
 where $R(\theta)$ is the curvature of the connection $\theta$
 \[ R(\theta)_a{}^b = \dd \theta_a{}^b + \theta_c{}^b\wedge\theta_a{}^c~.
 \]
 Then
 \begin{align*}
 \dth^2\, y_a - ({\cal F}^2 -  {\cal \widetilde R}{}^2)(y)_a 
 &=
  y^c\wedge \Big(- R(\theta)_{ac}
 \\[2pt]
 &~
 + \frac{\alpha'}{4}\, \big(
 \tr (F_{ab}\, \dd x^b\wedge F_{cd}\, \dd x^d)
 - \tr (\tilde R_{ab}\, \dd x^b\wedge \tilde R_{cd}\, \dd x^d)
 \big)
 \Big)
 \end{align*}
 By the Bianchi identity of the anomaly cancelation condition \eqref{eq:BIanom}, we have that
 \[
 (\dd H)_{abcd}\, \dd x^{bd} = \alpha'\,
 \big(
 \tr (F_{ab}\, \dd x^b\wedge F_{cd}\, \dd x^d - F_{ac}\, F)
 - \tr (\tilde R_{ab}\, \dd x^b\wedge \tilde R_{cd}\, \dd x^d - \tilde R_{ac}\, \tilde R)
 \big)~,
\]
which implies
\begin{align*}
 \dth^2\, y_a - ({\cal F}^2 -  {\cal \widetilde R}{}^2)(y)_a 
 &=
  y^c\wedge \Big(- R(\theta)_{ac} + \frac{1}{4}\, (\dd H)_{abcd}\, \dd x^{bd}
 \\[5pt]
 &\qquad\qquad\qquad + \frac{\alpha'}{4}\, (\tr (F_{ac}\, F)  -  \tr(\tilde R_{ac}\, \tilde R))
   \Big)
 \end{align*}
 To prove equation \eqref{eq:nil11} we contract this result with $\varphi$ to find
 \begin{equation}
 (\dth^2\, y_a + ({\cal F}^2 -  {\cal \widetilde R}{}^2)(y)_a)\lrcorner\varphi
= - y^b\lrcorner 
\left(\Big( R(\theta)_{ab} + \frac{1}{4}\, (\dd H)_{abcd}\, \dd x^{cd}\Big)\lrcorner\varphi \right)
= 0~,
 \end{equation}
 by propositions in the appendix  \ref{sec:appcurvs}.

\end{proof}

This result is certainly very interesting  and leads to an equally interesting corollary. 
As an exterior covariant derivative defined on $\cal Q$, one can write $\cal D$ in terms of a one form connection $\cal A$ on $\cal Q$ so that
 \[ {\cal D} = \dd_{\cal A} = \dd + {\cal A}~.\]
Then theorem \ref{theorem:Dinstanton} is equivalent to the statement that for a heterotic $G_2$ system  
\begin{equation*}
F({\cal A})\wedge\psi=0\:,
\end{equation*}
where $F({\cal A})=\dd{\cal A}+{\cal A}\wedge{\cal A}\in \Lambda^2(Y, {\rm End}({\cal Q}))$ is the curvature of ${\cal A}$. In other words, the connection one form ${\cal A}$ defines an instanton connection on $\cal Q$.

\subsection{The infinitesimal deformations of heterotic $G_2$ systems}

Consider the action of $\cal D$ on $p$-forms with values in $Q$
\begin{align*}
{\cal D}
\left( 
\begin{matrix}
y\\ \kappa \\ \alpha
\end{matrix}
\right)
= 
\left( 
\begin{matrix}
\dth \, y + {\cal \widetilde R}(\kappa) - {\cal F}(\alpha)\\
\dd_{\tilde\theta}\kappa + {\cal \widetilde R}(y) \\
\dd_A\alpha + {\cal F}(y)
\end{matrix}
\right)
\end{align*}
The idea is to construct a differential operator $\cal D$ is such that $\check{\cal D}$-closed one forms with values in $Q$ give the equations for infinitesimal moduli of heterotic $G_2$ systems. Let $\cal D$ act on an element
\[ {\cal Z} = \left( 
\begin{matrix}
y_t\\ \kappa_t \\ \alpha_t
\end{matrix}
\right)\quad \in \Lambda^1(Y,{\cal Q)}~.
\] 
Then
\begin{equation*}
\check{\cal D}{\cal Z}
= 0
\end{equation*}
if and only if
\begin{align}
\cdth \, y_t + {\check{\cal \widetilde R}}(\kappa_t) - \check{\cal F}(\alpha_t) &= 0~,
\label{eq:modfull}
\\
\check\dd_{\tilde\theta}\kappa_t + {\check{\cal \widetilde R}}(y_t) &= 0~,
\label{eq:modTY}
\\
\check\dd_A\alpha_t + \check{\cal F}(y_t) &= 0
~.\label{eq:modbundle}
\end{align}
In these equations $y_t$ is a general one form with values in $T^*Y$. 
To relate these equations with those equations for moduli we have obtained in sections \ref{sec:42} and \ref{sec:instantons}, we
set 
\begin{equation}
y_{t\, a} = M_{t\, a} + z_{t\, a}~,
\label{eq:splity}
\end{equation}
where the one form $z$ with values in $T^*Y$ corresponds to a two form
\[ z_t  = \frac{1}{2}\, z_{t\, ab}\, \dd x^{ab} \in \Lambda_{14}^2(Y) ~,\]
and where the antisymmetric part of the $7\times 7$ matrix associated to $M_t$ 
forms  a two form $m_t\in \Lambda_7^2(Y)$.

Consider first equation \eqref{eq:modbundle}. Using equation \eqref{eq:splity} we
have
\[
0 = \check\dd_A\alpha_t + \check{\cal F}(y_t)
= \check\dd_A\alpha_t + \check{\cal F}(M_t) + \check{\cal F}(z_t)
~.\]
However, the last term vanishes by equation \eqref{eq:i14to14}, giving
\[
\check\dd_A\alpha_t + \check{\cal F}(M_t) = 0~.
\]
By identifying $M_t$ precisely with one forms in $T^*Y$ corresponding to deformations of the
$G_2$ structure $\partial_t\varphi$ as in equation \eqref{eq:varphi2}, 
we obtain equation \eqref{eq:AtiyahBis}.  This equation
gives the simultaneous deformations of $(Y,V)$ that preserve the integrable $G_2$ structure on $Y$ and the instanton constraint on $V$.
 Note that we have no freedom in this identification.   There is of course an analogous discussion for equation \eqref{eq:modTY}.
 
 Consider now equation \eqref{eq:modfull}.
We have
\begin{equation}
\dth \, y_{t\, a} + {\cal \widetilde R}(\kappa)_{t\, a} - {\cal F}(\alpha)_{t\, a}
= \dth \, y_{t\, a}
-  \frac{\alpha'}{4}\, \big(
\tr (\alpha_t\wedge F_{ab}\, \dd x^b) - \tr (\kappa_t\wedge \tilde R_{ab}\, \dd x^b) \big)~.
\label{eq:modmaybe1}
\end{equation}
This equation should be identified with the results in section \ref{sec:finitemod}.  To do so we need the variations of anomaly cancelation condition.

 \begin{proposition} 
  Let $\alpha_t\in\Lambda^1({\rm End}(V))$ and $\kappa_t\in\Lambda^1({\rm End}(TY))$ correspond, respectively, to covariant variations of the connections $A$ and $\tilde\theta$ (see equation \eqref{eq:covvarA}).
  The variation of equation \eqref{eq:Hdef} can be written as
 \begin{align}
 \partial_t H&= \dd {\cal B}_t + \frac{\alpha'}{2}\, ({\rm tr}(F\wedge\alpha_t) - {\rm tr}(\tilde R\wedge\kappa_t))
 \label{eq:defanom}
 ~,
 \end{align}
 where ${\cal B}_t$ is a well defined 2 form, that is, it is invariant under gauge transformations of the bundles
 $V$ and $TY$ \footnote{This proposition is a generalisation to the $G_2$ case of the considerations in \cite{delaOssa:2014cia} and \cite{Candelas:2016usb} where an invariant variation of the $B$ field was studied in the context of heterotic compatifications on six dimensional manifolds.  The proof is of course identical.}.  In this definition $\Lambda_t$ is a connection on the moduli space of instanton bundles on $V$ and  $\tilde\Lambda_t$ is a connection on the moduli space of instanton bundles on $TY$ (see discussion in section \ref{sec:instantons}). 
 \end{proposition}
\begin{proof}
Consider the variations of \eqref{eq:Hdef}.  We compute first the variations of the Chern-Simons term for the gauge connection.
\begin{equation*}
 \partial_t {\cal CS}[A] = {\rm tr}\left(- \dd(A\wedge \partial_t A) + 2\, F\, \partial_t A\right)~,
 \end{equation*}
 and therefore
 \begin{equation}
 \begin{split}
 \partial_t H = &\dd\left(
 \partial_t B - \frac{\alpha'}{4}\, ({\rm tr}(A\wedge\partial_t A) - {\rm tr}(\tilde\theta\wedge\partial_t \tilde\theta))\right)
 \\[5pt] 
 &\qquad\qquad\qquad\qquad\qquad\qquad+ \frac{\alpha'}{2}\, ({\rm tr}(F\wedge\partial_t A) - {\rm tr}(\tilde R\wedge\partial_t \tilde\theta))
 \label{eq:prevarH}
 \end{split}
 \end{equation}
 To obtain the desired results we replace $\partial_t A$ and $\partial_t\tilde\theta$ with $\alpha_t$ and $\kappa_t$ at the expense of introducing connections $\Lambda_t$ and $\tilde\Lambda_t$ on the moduli space of instanton bundles on $V$ and $TY$ respectively as explained in section \ref{sec:instantons}.  We have for the second term in equation \eqref{eq:prevarH}
\[
\frac{\alpha'}{2}\, ({\rm tr}(F\wedge\partial_t A) 
= \frac{\alpha'}{2}\, ({\rm tr}(F\wedge(\alpha_t + \dd_A\Lambda_t))
= \frac{\alpha'}{2}\, ({\rm tr}(F\wedge\alpha_t) +  \frac{\alpha'}{2}\, \dd\tr(F\, \Lambda_t)~,
\]
where we have used equation \eqref{eq:covvarA} and the Bianchi identity $\dd_A F = 0$. A similar relation is obtained for the third term of equation \eqref{eq:prevarH}. 
Then equation \eqref{eq:prevarH} gives equation \eqref{eq:defanom}
\[
\partial_t H = \dd {\cal B}_t  + \frac{\alpha'}{2}\, ({\rm tr}(F\wedge\alpha_t) - {\rm tr}(\tilde R\wedge\kappa_t))~,
\]
where we have defined ${\cal B}_t $ such that
\[
\dd {\cal B}_t = \dd\left(
 \partial_t B - \frac{\alpha'}{4}\, (
 {\rm tr}(A\wedge\partial_t A  - 2 F\, \Lambda_t) 
 - {\rm tr}(\tilde\theta\wedge\partial_t \tilde\theta - 2 \tilde R\, \tilde\Lambda_t)
 )\right)~. 
 \]
 Note that, as both $\partial_t H$ and the second term in equation \eqref{eq:defanom} are gauge invariant,  then so is $\dd{\cal B}_t $.
 We can now manipulate this result to obtain
 \[
\dd {\cal B}_t = \dd\left(\partial_t B 
 - \frac{\alpha'}{4}\Big(
 {\rm tr}\big(A\wedge \alpha_t - \Lambda_t\,\dd A - \dd(A\Lambda_t)\big)
  - {\rm tr}\big(\tilde\theta\wedge \kappa_t - \tilde\Lambda_t\,\dd \tilde\theta - \dd(\tilde\theta\tilde\Lambda_t)\big)
 \Big)
 \right)~.
 \]
 
\end{proof}
In our considerations below, the explicit form of ${\cal B}_t$ is not needed.  However it is important to keep in mind that ${\cal B}_t$ is defined
{\it up to a gauge invariant closed form} leading to an extra symmetry of heterotic $G_2$ systems.  We discuss the meaning of this symmetry below.

Returning to equation \eqref{eq:modmaybe1}, using equation \eqref{eq:defanom} 
we have that
\begin{align*}
\frac{1}{4}\,\big(\dd {\cal B}_t - \partial_t H\big)_{abc}\, \dd x^{bc} 
 &= -\frac{\alpha'}{8}\,3 \,\big( \tr ( \alpha_{t\, [a}\, F_{bc]}) 
- \tr(\kappa_{t\, [a}\, \tilde R_{bc]})\big)\, \dd x^{bc}
\\[5pt]
&= \frac{\alpha'}{4}\,\big( \tr ( - \alpha_{t\, a}\, F + \alpha_t\wedge F_{ab}\, \dd x^b) 
- \tr(- \kappa_{t\, a}\, \tilde R +  \kappa_t\wedge \tilde R_{ab}\, \dd x^b)\big)~,
\end{align*}
which implies
\[\frac{\alpha'}{4}\,\big( \tr (\alpha_t\wedge F_{ab}\, \dd x^b) 
- \tr(\kappa_t\wedge \tilde R_{ab}\, \dd x^b)\big)
= \frac{1}{4}\,\big(\dd {\cal B}_t - \partial_t H\big)_{abc}\, \dd x^{bc}
+ \frac{\alpha'}{4}\,\big( \tr ( \alpha_{t\, a}\, F) 
- \tr(\kappa_{t\, a}\, \tilde R)\big)~. 
\]
Using this result into the right hand side of equation \eqref{eq:modmaybe1} we find
\begin{align*}
\dth \, y_{t\, a} + {\cal \widetilde R}(\kappa)_{t\, a} - {\cal F}(\alpha)_{t\, a}
&= \dth \, y_{t\, a}
- \frac{1}{4}\,\big(\dd {\cal B}_t - \partial_t H\big)_{abc}\, \dd x^{bc}
\\[5pt]
&\qquad \qquad \qquad \qquad \qquad - \frac{\alpha'}{4}\,\big( \tr ( \alpha_{t\, a}\, F) 
- \tr(\kappa_{t\, a}\, \tilde R)\big)
~.
\end{align*}
Contracting this with $\varphi$ we find
\begin{align*}
0 &= (\dth \, y_{t\, a} + {\cal \widetilde R}(\kappa)_{t\, a} - {\cal F}(\alpha)_{t\, a})\lrcorner\varphi
\nn\\
&= (\dth \, y_{t\, a})\lrcorner\varphi
- \frac{1}{4}\,\big(\dd {\cal B}_t - \partial_t H\big)_{abc}\, \varphi^{bc}{}_d\dd x^d ~,
\end{align*}
which can be written equivalently as
\be
0= \left(\dth \, y_{t\, a} + \frac{1}{2} \, \big(- (\dd{\cal B}_t)_a + \partial_t\hat H_a\big)\right)\lrcorner\varphi~.
\label{eq:genmod}
\ee
This result needs to be consistent with the analysis of the moduli of integrable $G_2$ structures.  We recall that in section \ref{sec:finitemod} we obtained instead
\[
0 = \left(\dth M_{t \, a} + (\dd m_t)_a + \frac{1}{2}\, \partial_t\hat H_a\right)\lrcorner\varphi~,
\]
where there  $\pi_{14}(m_t)$ drops out of this equation.  To be able to compare these equations, we use
\eqref{eq:splity} in equation \eqref{eq:genmod} and  we now have
\[
0 = \left(\dth \, (M_t + z_t)_a + \frac{1}{2} \, \big(- (\dd{\cal B}_t)_a + \partial_t\hat H_a\big)\right)\lrcorner\varphi
~,
\]
which by  lemma \ref{lemma:lemz} gives
\[
0 = \left(\dth \, M_{t\, a} + \frac{1}{2} \, \big(- (\dd(2 \, z_t +{\cal B}_t))_a + \partial_t\hat H_a\big)\right)\lrcorner\varphi
~.
\] 
Therefore we find
\begin{equation*}
0 = (\dd ( 2(m_t + z _t) - {\cal B}_t))_a\lrcorner \varphi~,
\end{equation*}
which implies
\begin{equation}
\dd ( 2(m_t + z _t) - {\cal B}_t)) = 0~.
\label{eq:zbeta}
\end{equation}
This equation {\it identifies} the degrees of freedom corresponding to the antisymmetric part of $y_t$, that is $m_t + z_t$, with 
the invariant variations of the $B$ field as follows
\be 2\, (z _t + m_t) + \mu_t = {\cal B}_t 
~, \label{eq:zbeta2}
\ee
where $\mu_t$ is a gauge invariant $\dd$-closed two form. This ambiguity in the definition of ${\cal B}_t $ has already been noted above. With this identification, we conclude that $\cal D$ is such that $\check{\cal D}$-closed one forms with values in $Q$ correspond to infinitesimal moduli of the heterotic vacua.

 \subsection{Symmetries and trivial deformations}

Let us now discuss trivial deformations. On the one hand, these should have an interpretation in terms of symmetries of the theory, {\it i.e.}diffeomorphisms and gauge transformations of $A$, $\theta$ and $B$. On the other hand, since $\check{\cal D}^2 = 0$, trivial deformations are given by 
\[
{\cal Z}_{\rm triv} = \check{\cal D} {\cal V} ~,
\]
where ${\cal V} =(v,\pi,\epsilon)^T$ is a section of ${\cal Q} = T^*Y\,  \oplus\, {\rm End}(TY)\oplus\, {\rm End}(V)$. We show that ${\cal Z}_{\rm triv}$ can indeed be interpreted in terms of symmetries of the theory:
 \[
{\cal Z}_{\rm triv} = \left( \begin{matrix}
y_{\rm triv}\\ \kappa_{\rm triv} \\ \alpha_{\rm triv}
\end{matrix}
\right)
= \left( \begin{matrix}
\dth \, v + {\cal \widetilde R}(\pi) - {\cal F}(\epsilon)\\
\dd_{\tilde\theta}\pi + {\cal \widetilde R}(v) \\
\dd_A\epsilon + {\cal F}(v)
\end{matrix}
\right) \, .
\]
Let us start with the last entry of this vector, where the first term, $ \dd_A\epsilon$, corresponds to gauge transformations of the gauge field. To interpret the second term, note that under diffeomorphisms, $F$ transforms as
\[
{\cal L}_v F = v\lrcorner \dd_A F + \dd_A(v \lrcorner F) = \dd_A(v \lrcorner F) = \dd_A({\cal F}(v)) \, ,
\]
where we have used the definition of the map ${\cal F}$ given at the beginning of this section.  Thus,  the second term corresponds to the change of the gauge field $A$ under diffeomorphism. Analogously, we may interpret $\dd_{\tilde\theta}\pi$ as a gauge transformation, and ${\cal \widetilde R}(v)$ as a diffeomorphism, of the connection $\tilde{\theta}$ on the tangent bundle.

We move on to show that 
\[
y_{{\rm triv} \, a}=\dth \, v_a + {\cal \widetilde R}(\pi)_a - {\cal F}(\epsilon)_a \; 
\]
corresponds to trivial deformations of the metric and $B$-field.
Thinking of $y_{\rm triv \, ab}$ as a matrix, the symmetric part corresponds to
\[
y_{{\rm triv} \, (ab)} = \dd_{\theta (b} \, v_{a)} = \nabla^{LC}_{(a} v_{b)} \, .
\]
Comparing with equation \eqref{eq:trivsym} and \eqref{eq:metricvar} (for more details see  proposition 3 and theorem 8 of \cite{delaOssa:2016ivz}),  one concludes that these are trivial deformations of the metric. For the antisymmetric part, it is useful to define a two-form 
\be \label{eq:ytrivas}
\begin{split}
y_{\rm triv}^{\rm antisym} &\equiv\frac{1}{2}y_{{\rm triv} \,[ab]} \dd x^{ab}
\\ &=
 \frac{1}{2}( \dth \, v_{a})_b \dd x^{ab} - \frac{\alpha'}{4} \left({\rm tr} [\epsilon F] -   {\rm tr} [\pi {\tilde R}] \right)\, 
  \\ &=
   \frac{1}{2}( \partial_b v_a - \Gamma_{ab}{}^c v_c ) \dd x^{ab}
 - \frac{\alpha'}{4} \left({\rm tr} [\epsilon F] -   {\rm tr} [\pi {\tilde R}] \right) \, 
 \\ &=
  - \frac{1}{2}( \dd v + v\lrcorner H ) 
  - \frac{\alpha'}{4} \left({\rm tr} [\epsilon F] -   {\rm tr} [\pi {\tilde R}] \right) \, .
 \end{split}
\ee
This equation should be equivalent to
\be \label{eq:ytrivas2}
y_{\rm triv}^{\rm antisym} = \frac{1}{2} ({\cal B}_{\rm triv}  - \mu_{\rm triv}) \; ,
\ee
as is required by \eqref{eq:splity} in combination with \eqref{eq:zbeta2}. To prove this we must specify what ${\cal B}_{\rm triv}$ and $\mu_{\rm triv}$ are. The latter is simple: since $\mu_{\rm t}$ is a closed two-form, $\mu_{\rm triv}$ must be exact. Physically, $\mu_{\rm triv}$ corresponds to  a gauge transformation of $B$ (this gauge transformation is not to be confused with gauge transformations of the bundles).

We may determine ${\cal B}_{\rm triv}$ by requiring that it corresponds to changes in the physical fields $B, A$ and $\tilde{\theta}$ that at most change $H$ by a diffeomorphism. Concordantly, we compare $\partial_{\rm triv} H$ from \eqref{eq:defanom} 
\[ \begin{split}
\partial_{\rm triv} H
&= \dd {\cal B}_{\rm triv} + \frac{\alpha'}{2}\, ({\rm tr}[F\wedge \alpha_{\rm triv}] - {\rm tr}[\tilde R\wedge\kappa_{\rm triv}]) \\
&= \dd {\cal B}_{\rm triv} + \frac{\alpha'}{2}\, ({\rm tr}[F\wedge (\dd_A\epsilon + {\cal F}(v))] - {\rm tr}[\tilde R\wedge(\dd_{\theta}\pi + \tilde{\cal R}(v))])\\
&= \dd \left( {\cal B}_{\rm triv} + \frac{\alpha'}{2}\, \left({\rm tr}[F\epsilon] - {\rm tr}[\tilde R\pi]\right)\right)- \frac{\alpha'}{2} \left( 
{\rm tr}[{\cal F}(v) \w F] - {\rm tr}[\tilde{{\cal R}}(v) \w \tilde{R}]
\right)
\end{split}
\]
with the Lie derivative of $H$:
\[ \begin{split}
{\cal L}_v H &= v\lrcorner \dd H + \dd(v \lrcorner H)
\\&= \frac{\alpha'}{4} v\lrcorner \left( 
{\rm tr}[F \w F] - {\rm tr}[\tilde{R} \w \tilde{R}]
\right)
+ \dd(v \lrcorner H)
\\&= \frac{\alpha'}{2} \left( 
{\rm tr}[g^{ab}v_a F_{bc} \dd x^c \w F] - {\rm tr}[g^{ab}v_a \tilde{R}_{bc} \dd x^c \w \tilde{R}]
\right)
+ \dd(v \lrcorner H)
\\&= \frac{\alpha'}{2} \left( 
{\rm tr}[{\cal F}(v) \w F] - {\rm tr}[\tilde{{\cal R}}(v) \w \tilde{R}]
\right)
+ \dd(v \lrcorner H) \, .
 \end{split}
\]
We find that trivial transformations of $H$ correspond to a diffeomorphism
\[
\partial_{\rm triv} H = {\cal L}_{-v} H
\]
provided that 
\[
{\cal B}_{\rm triv} = - v \lrcorner H - \frac{\alpha'}{2}\, \left( {\rm tr}(F\epsilon) - {\rm tr}(\tilde R\pi)\right) \; ,
\]
up to a closed two form. 
Inserting this in \eqref{eq:ytrivas2}, we thus reproduce \eqref{eq:ytrivas}. If follows that $y_{\rm triv}^{\rm antisym}$ corresponds to gauge transformations and diffeomorphisms of $H$. This concludes the proof that  ${\cal Z}_{\rm triv}$ can be interpreted in terms of symmetries of the theory.

\subsection{The tangent space to the moduli space and $\alpha'$ corrections}
\label{sec:Acorr}

We have shown so far that the tangent space $\cal TM$ to the moduli space $\cal M$ of heterotic $G_2$ structures $[ (Y, \varphi), (V, A), (TY,\tilde\theta),H]$ is given by
\[
{\cal TM} = H^1_{\check{\cal D}}(Y, {\cal Q})~,\]
where ${\cal D}$ is a covariant exterior derivative given in \eqref{eq:calD} which satisfies $\check{\cal D}^2 = 0$, or equivalently, the bundle $\cal Q$ has an instanton connection $\cal A$ such that 
\[ {\cal D} = \dd_{\cal A} = \dd + {\cal A}~.\] 

To close our analysis of the infinitesimal deformations of heterotic $G_2$ systems, we discuss how $\alpha'$ corrections might modify the results obtained above. In theorem \ref{theorem:Dinstanton} we have assumed that the connections $A$ and $\tilde\theta$ are instanton connections
on $V$ and $TY$ respectively, which we know to be true to first order in $\alpha'$.  
We want to see what happens when we relax these conditions. We note first that our discussion concerning the moduli of heterotic compactifications on integrable $G_2$ manifolds is accurate from a physical perspective to ${\cal O}(\alpha'^2)$, provided the connection $\dd_{\tilde\theta}$ satisfies the instanton condition \cite{delaOssa:2014msa}. The naturalness of the structure however makes it very tempting to conjecture that the analysis holds to higher orders in $\alpha'$ as well, as is also expected in compacifications to four dimensions \cite{delaOssa:2014msa, Candelas:2016usb, McOrist:2016cfl}. A detailed analysis of higher order $\alpha'$ effects is beyond the scope of the present paper.  However, in the following theorem we find a remarkable result, which amounts to the converse of theorem \ref{theorem:Dinstanton}, in particular the Bianchi identity of the anomaly cancelation condition is  {\it deduced} from the requirement that the operator $\cal D$ defined by equation \eqref{eq:calD} satisfies the condition $~\check{\cal D}^2= 0$.

\begin{theorem}
\label{tm:alpha}
Let $Y$ be a manifold with a $G_2$ structure, $V$ a bundle on $Y$ with connection $A$, and $TY$ the tangent bundle of $Y$ with connection $\tilde\theta$. Let $\theta$ be a metric connection compatible with the $G_2$ structure, that is $\nabla\varphi=0$ with connection symbols  $\Gamma$ such that
$\theta_a{}^b = \Gamma_{ac}{}^b\, \dd x^c$.  
Consider the exterior derivative $\cal D$ defined by equation \eqref{eq:calD} and assume that  $~\check{\cal D}^2= 0$. Then 
$\big( [Y,\varphi],[V,A],[TY,\tilde{\theta}],H \big)$ is a heterotic system.  This statement is true to all orders in the perturbative $\alpha'$ expansion.
\end{theorem}
\begin{proof}
Consider again equation \eqref{eq:nilD} and assume now that  $\check{\cal D}^2= 0$.  We use the $\alpha'$ expansion to prove this theorem. 

 We begin with the $(33)$ entry of equation
\eqref{eq:nilD}, that is assume first that
\be {\check\dd}^2_A(\alpha) - \check{\cal F}^2(\alpha) = [\pi_7(F), \alpha] - \check{\cal F}^2(\alpha) = 0 ~,
\label{eq:iterative}
\ee
for all $\alpha\in \Lambda^p(Y, {\rm End}(V))$.  Because ${\cal F}^2$ is of order $\alpha'$ (see equation \eqref{eq:calFsquared}), it must be the case that 
$\pi_7(F)$ is at least of order $\alpha'$. Therefore, $F\in\Lambda_{14}^2(Y, {\rm End}(V))$ modulo $\alpha'$ corrections. By equation \eqref{eq:i14to14}, this in turn means for the second term in equation \eqref{eq:iterative}, that ${\check{\cal F}}^2(\alpha)=0$ modulo ${\cal O}(\alpha'^2)$, and hence the first term must also be ${\cal O}(\alpha'^2)$.  In other words, $F$ is in the $\bf 14$ representation modulo $\alpha'^2$ corrections.
Employing \eqref{eq:i14to14}, again we see that the second term  of \eqref{eq:iterative} is at least of ${\cal O}(\alpha'^3)$. Continuing this iterative procedure order by order in $\alpha'$ we find that
 \begin{equation}
\pi_7(F)=0\:.
\end{equation}
Therefore, the two terms of equation \ref{eq:iterative} vanish separately. In particular
$\check\dd_A^2 = 0$ if and only if $Y$ has an integrable $G_2$ structure and $A$ is an instanton connection  on $V$. 
The proof for the entry $(22)$ of  \eqref{eq:nilD} corresponding to the connection $\tilde\theta$ on $TY$ is similar, so $\tilde\theta$ is an instanton connection on $TY$.  With this result and the proof of theorem \ref{theorem:Dinstanton} all the other entries in \eqref{eq:nilD} vanish, except the entry $(11)$.

For the $(11)$ entry of \eqref{eq:nilD}, we now assume that
\[
\cdth^2\, y+ {\check{\widetilde{\cal R}}}{}^2(y) - \check{\cal F}^2(y) = 
0~,
\]
for all $y\in \Lambda^p(Y, T^*Y)$.  This is equivalent to 
\[
\Big(- R(\theta)_{ab}
 + \frac{\alpha'}{4}\, \big(
 \tr (F_a\, \wedge F_b)
 - \tr (\tilde R_a\,\wedge \tilde R_b)
 \big)
 \Big)\lrcorner\varphi = 0 ~,
\]
As the $G_2$ structure is integrable, we take $\nabla$ to be a connection with totally antisymmetric torsion $H$ (see equations \eqref{eq:dphi} and \eqref{eq:dpsi}).  This together with the identity
\[
R(\theta)_{ab}\lrcorner\varphi = - \frac{1}{4}\, (\dd H)_{cdab}\, \varphi^{cd}{}_e\, \dd x^e
\]
 in appendix \ref{sec:appcurvs}, gives
\begin{align*}
0 &= (\dd H)_{cdab}\, \varphi^{cd}{}_e\, \dd x^e
 + \alpha'\, \big(
 \tr (F_a\, \wedge F_b)
 - \tr (\tilde R_a\,\wedge \tilde R_b)
 \big)
 \lrcorner\varphi
 \\
 &=  \left((\dd H)_{cdab}
 + \alpha'\, \big(
 \tr (F_{ac}\, F_{bd})
 - \tr (\tilde R_{ac}\,\tilde R_{bc})
 \big)\right)\, \varphi^{cd}{}_e\, \dd x^e
 \\[5pt]
 &=  \left((\dd H)_{cdab}
 - 3\, \frac{\alpha'}{2}\, \big(
 \tr (F_{[cd}\, F_{ab]})
 - \tr (\tilde R_{[cd}\,\tilde R_{ab]})
 \big)\right)\, \varphi^{cd}{}_e\, \dd x^e
 \end{align*}
 where in the last equality we have used the fact that both $A$ and $\tilde\theta$ are instantons.
 Then
 \be
 0 = \left(\dd H
 - \frac{\alpha'}{4}\, \big(
 \tr (F\wedge F)
 - \tr (\tilde R\wedge\tilde R)
 \big)\right)_{cdab}\, \varphi^{cd}{}_e\, \dd x^e~.
 \label{eq:preanomaly}
 \ee
Consider the four form
\[
\Sigma = \dd H - \frac{\alpha'}{4}\, \big(\tr(F\wedge F) - \tr(R\wedge F)\big)
\]
and the associated three form $\Sigma_a$ with values in $T^*Y$. Then equation \eqref{eq:preanomaly} is equivalent to $\Sigma_a = 0$ 
to and hence $\Sigma = 0$.  Note that, in this way we have also proved that the Bianchi identity of the anomaly cancelation condition does
not receive higher order $\alpha'$ corrections.
\end{proof}

We remark that Theorem \ref{tm:alpha} relies heavily on the $\alpha'$ expansion.  Mathematically, there is no reason to assume that such an expansion exists. It is tempting to speculate that the form of the covariant derivative $\cal D$ on $\cal Q$ is the correct operator including all quantum corrections, also the non-pertubative ones. This would imply that the quantum corrected geometry is encoded in an instanton connection on $\cal Q$ even if the connections $A$ and $\tilde\theta$ need not be instantons anymore.

\section{Conclusions and outlook}
\label{sec:concl}

This paper has been devoted to the analysis of the infinitesimal ``massless'' deformations of heterotic string compactifications on a seven dimensional compact manifold $Y$ of integrable $G_2$ structure. We have seen that the heterotic supersymmetry conditions together with the heterotic Bianchi identity can be put in terms of a differential $\check{\cal D}$ on a bundle ${\cal Q}=T^*Y\oplus\End(TY)\oplus\End(V)$. That is,
\begin{equation}
\check{\cal D}\::\;\;\;\check\Lambda^p({\cal Q})\rightarrow\check\Lambda^{p+1}({\cal Q})\:,\;\;\;\check{\cal D}^2=0\:,
\end{equation}
where $\check\Lambda^p({\cal Q})$ is an appropriate sub-complex of ${\cal Q}$-valued forms. Furthermore, the space of infinitesimal deformations of such compactifications is parametrised by
\begin{equation}
{\cal TM}=\check H^1_{\check{\cal D}}({\cal Q})\:,
\end{equation}
where ${\cal TM}$ denotes the tangent space of the full moduli space.

Our deformation analysis naturally incorporates fluctuations of the heterotic $B$-field. In fact, due to the anomaly cancellation condition, we could only translate the heterotic $G_2$ system into  $\check{\cal D}$-closed ${\cal Q}$-valued one-forms if these one-forms included $B$-field fluctuations. Put differently, to disentangle geometric and $B$-field deformations we must  decompose the one forms with values in $TY$ into two sets  ${\cal S}(TY)$ and  ${\cal A}(TY)$, that correspond to symmetric and  antisymmetric matrices respectively.  This decomposition does not serve to simplify the analysis of the deformation, and in fact seems unnatural from the perspective of ${\cal Q}$. We should remark that for the $G_2$ holonomy, the inclusion of  ${\cal A}(TY)$ among the infinitesimal moduli is natural but not necessary \cite{delaOssa:2016ivz}.

Another interesting point regards the ${\mathcal O}(\alpha')$ corrections to the $H$-flux Bianchi identity, which arise as a consequence of an anomaly cancellation condition in the world-sheet description of the heterotic string. We observe that these ${\mathcal O}(\alpha')$ corrections are really imposed already in our geometric analysis of the supergravity system, as a necessary constraint to obtain a good deformation theory. This provides an alternative argument for why the $\alpha'$ corrections of heterotic supergravity take the form observed by  Bergshoeff--de Roo \cite{Bergshoeff1989439}, that could be of use when deriving higher order corrections, without need of analysing the world sheet description of the string.

The deformations of heterotic $G_2$ systems are similar to the deformations of the six dimensional holomorphic Calabi--Yau and Strominger--Hull system as it appears in the papers \cite{Anderson:2009nt,Anderson:2011ty,Anderson:2014xha, delaOssa:2014cia, Garcia-Fernandez:2015hja, delaOssa:2015maa, Candelas:2016usb, McOrist:2016cfl}, though there are some notable differences. In particular, in contrast to the Atiyah-like holomorphic extension bundle of the Strominger--Hull system, $\check{\cal D}$ is not upper triangular with respect to the components of $\cal Q$, and hence $(\cal Q,D)$ does not form an extension bundle in the usual sense. This obscures some properties of the three-dimensional low-energy effective field theory, {\it i.e.} the relation between the massless spectrum and cohomology groups which exist in the holomorphic case. Extension bundles also fit naturally into the heterotic generalised geometry developed in reference \cite{Clarke:2016qtg} (see also \cite{Coimbra:2014qaa}). We leave it as an open question whether an analogue of Schur's lemma can be used to bring $\check{\cal D}$ to the required form, {\it i.e.} by projecting the complex $\check\Lambda^p({\cal Q})$ onto further sub-representations. Deeper investigations into the properties of the connection $\cal D$ and the corresponding structure group of $(\cal Q, D)$ may provide a better understanding of the theory, that could clarify some of the points mentioned here.

An interesting connection between the heterotic $G_2$ system and the six dimensional Strominger--Hull system arises by embedding the latter into the former. This implies that  the seven dimensional structure unifies the holomorphic constraints,  conformally balanced condition and the Yang-Mills conditions of the Strominger--Hull system. We plan to study this unification, and the insight it may bring to the deformations of the Strominger--Hull and other six-dimensional heterotic systems, in the future. 

We have determined the infinitesimal moduli of heterotic $G_2$ systems, and a natural next question concerns that of higher order deformations and obstructions. On physical grounds, it is expected that the finite deformations can be parametrised as solutions ${\cal X}$ of a Maurer--Cartan equation
\begin{equation}
\check{\cal D}{\cal X}+\frac{1}{2}[{\cal X},{\cal X}]=0\:,\;\;\;{\cal X}\in\check\Lambda^1({\cal Q})\:,
\end{equation}
for some differential graded Lie algebra (DGLA). What exactly the Lie bracket 
\begin{equation}
[\:,\:]\::\;\;\;\check\Lambda^p({\cal Q})\times\check\Lambda^q({\cal Q})\rightarrow\check\Lambda^{p+q}({\cal Q})\:,
\end{equation}
and the corresponding DGLA is remains to be determined.\footnote{We expect that there exist a parametrisation where the deformation problem is governed by a DGLA, but from a mathematical standpoint this is not guaranteed. The deformations might instead be described by an $L_\infty$-algebra, including non-vanishing Jacobi identities and higher brackets.} In this paper we have laid the foundations for further investigations into such finite deformations, and we plan to exploit this groundwork in a future publication \cite{delaOssa17}.

\section*{Acknowledgements}
The authors would like to thank Anthony Ashmore, Andreas Braun, Philip Candelas, Marc-Antoine Fiset, Mario Garc\'ia-Fern\'andez, Chris Hull, Spiro Karigiannis, Ulf Lindstr\"om, Jock McOrist, Ruben Minasian, George Papadopoulos, Michela Petrini, Brent Pym, Carlos Shahbazi and Carl Tipler for interesting discussions. ML and XD thank the Mainz Institute for Theoretical Physics (MITP) for hospitality when part of this work was completed.  ML also acknowledges partial support from the COST Short Term Scientific Mission MP1210-32976, and thanks Oxford University for hospitality.  EES thanks the CERN short visitor program, and the CERN Theory Group for hospitality when part of this work was completed.

The work of XD is supported in part by the EPSRC grant EP/J010790/1. ML's research is financed by the Swedish Research Council (VR) under grant number 2016-03873. The work of EES, made within the Labex Ilp (reference Anr-10-Labx-63), was supported by French state funds managed by the Agence nationale de la recherche, as part of the programme Investissements d'avenir under the reference Anr-11-Idex-0004-02. 

This is a pre-print of an article published in Communications in Mathematical Physics. The final authenticated version is available online at:\\ https://doi.org/10.1007/s00220-017-3013-8.

\newpage
\appendix
\section{Identities and Lemmas}\label{app:formulas}

We have used a number of identities in this paper, and collect some of them in this appendix. Many of these formulas can be found in the literature, {\it e.g.}  \cite{Bryant:2005mz}, and further relevant formulas can be found in {\it e.g.}  \cite{delaOssa:2016ivz} and \cite{delaOssa:2014lma}.

The operator $\lrcorner$ denotes the contraction of forms, and is defined by
\be \alpha\lrcorner\beta = 
\frac{1}{k!\, p!}\ \alpha^{m_1\cdots m_k}\ \beta_{m_1\cdots m_k n_1\cdots n_p}
\dd x^{n_1}\cdots \dd x^{n_p}~,  \ee
where $\alpha$ is any $k$-form and $\beta$ is any $p+k$-form. It is easy to deduce the identity
\be \alpha\lrcorner \beta = (-1)^{p(d-p-k)}\, *(\alpha\wedge*\beta)~.\ee
For odd $d$ we have
\be \alpha\lrcorner \beta = (-1)^{pk}\, *(\alpha\wedge*\beta)~.\ee
\vskip10pt

\noindent
Contractions between $\varphi$ and $\psi$ give  \cite{Bryant:2005mz}
\begin{align}
\varphi^{abc}\, \varphi_{abc} &=42~,
\\
\varphi^{acd}\, \varphi_{bcd} &= 6\, \delta_b^a~,
\label{eq:g2ident2}\\
\varphi^{eab}\, \varphi_{ecd} &= 2\, \delta_{[c}^a\, \delta_{d]}^b
+ \psi^{ab}{}_{cd}~.
\label{eq:g2ident3}
\\[7pt]
\varphi^a{}^{d_1 d_2}\, \psi_{bc d_1 d_2} &= 4\, \varphi^a{}_{bc}~,
\label{eq:g2ident4}\\
\varphi^{abf}\, \psi_{cdef} &= - 6\, \delta^{[a}{}_{[c}\, \varphi^{b]}{}_{de]}~,
\label{eq:g2ident5}
\\[7pt]
\psi^{abcd}\psi_{abcd} &= 7\cdot 24 = 168~,
\\
\psi^{acde}\psi_{bcde} &= 24\, \delta_b^a~,\\
\psi^{abe_1e_2}\psi_{cde_1e_2} &= 8\, \delta_{[c}^a\, \delta_{d]}^b + 2\,\psi^{ab}{}_{cd}~,\\
\psi^{a_1a_2a_3c}\psi_{b_1b_2b_3c} &= 6\, \delta_{[b_1}^{a_1}\, \delta_{b_2}^{a_2} \, \delta_{b_3]}^{a_3}
+ 9\, \psi^{[a_1a_2}{}_{[b_1b_2}\, \delta^{a_3]}_{b_3]} - \varphi^{a_1a_2a_3}\,\varphi_{b_1b_2b_3}~,
\\
\psi^{a_1a_2a_3a_4}\psi_{b_1b_2b_3b_4} &= 24\, \delta_{[b_1}^{a_1}\, 
\delta_{b_2}^{a_2} \, \delta_{b_3}^{a_3}\, \delta_{b_4]}^{a_4}
\\[2pt]
&\qquad
+ 72\, \psi^{[a_1a_2}{}_{[b_1b_2}\, \delta_{b_3}^{a_3}\,\delta^{a_4]}_{b_4]}- 16\,  \varphi^{[a_1a_2a_3}\,\varphi_{[b_1b_2b_3}\, \delta^{a_4]}_{b_4]}~,
\\[10pt]
\frac{\sqrt{g}}{2}\, \varphi_{acd}\, \epsilon^{cdb_1 b_2 b_3 b_4 b_5}
&= 5\, \delta_a^{[b_1}\, \psi^{b_2 b_3 b_4 b_5]}~,\label{eq:g2phiep}
\\[5pt]
\frac{\sqrt{g}}{3!}\, \psi_{ac_1c_2c_3}\, \epsilon^{c_1c_2c_3b_1 b_2 b_3 b_4}
&= -4 \, \delta_a^{[b_1}\, \varphi^{b_2 b_3 b_4]}~.\label{eq:g2psiep}
\end{align}
Let $\alpha$ be a one form (possibly with values in some bundle)
\begin{align}
\varphi\lrcorner(\alpha\wedge\varphi) &= (\alpha\lrcorner\psi)\lrcorner\psi = -4\, \alpha~, 
\label{eq:onepsipsi}\\
\psi\lrcorner(\alpha\wedge\psi) &= (\alpha\lrcorner\varphi)\lrcorner\varphi =3\, \alpha~, 
\label{eq:onephiphi}\\
\varphi\lrcorner(\alpha\wedge\psi) &= 
(\alpha\lrcorner \varphi)\lrcorner \psi
= 2\, \alpha\lrcorner\varphi~.
\label{eq:onephipsi}
\end{align}

\noindent
Let $\alpha$ be a two form (possibly with values in some bundle)
\begin{align}
\varphi\lrcorner(\alpha\wedge\varphi) &= -\, (\alpha\lrcorner\psi)\lrcorner\psi = 2\, \alpha + \alpha\lrcorner\psi~,
\label{eq:twopsipsi}\\
\psi\lrcorner(\alpha\wedge\psi) &= (\alpha\lrcorner\varphi)\lrcorner\varphi = 3\, \pi_7(\alpha) = 
\alpha + \alpha\lrcorner\psi~,
\label{eq:twophiphi}\\[3pt]
(\alpha\lrcorner\varphi)\lrcorner\psi &= \frac{1}{2}\, \alpha_a{}^d\, \varphi_{bcd}\, \dd x^{abc}~.
\label{eq:twophipsi}
\end{align}

\noindent
Let $\alpha$ be a two form  in $\Lambda^2_{14}(Y)$ (possibly with values in some bundle). We have
\be
\alpha\wedge\psi_a  = - \alpha_a\wedge\psi~.
\label{eq:idtwo14}
\ee

\subsection*{Useful Lemmas}
In the main part of the paper we have used some formula's without proof  in order to ease the flow of the text. Here we proove some of the relevant formulas, collected in a couple of lemmas.

\begin{lemma}
Let $\lambda\in\Lambda^3_{27}$. Then
\begin{align}
*\lambda &= - \frac{1}{4}\, \lambda_{ab}{}^e\,  \varphi_{ecd}\, \dd x^{abcd}~,
\label{eq:star327one}
\\[5pt]
\lambda_{a[bc}\, \psi_{def]}{}^a&= 0~.
\label{eq:star327two}
\end{align}
\end{lemma}

\begin{proof}
\begin{align*}
*\left(- \frac{1}{4}\, \lambda_{ab}{}^e\,  \varphi_{ecd}\, \dd x^{abcd}\right)
&= \frac{\sqrt{g}}{4!}\, \lambda^{abf}\, \varphi_f{}^{cd}\, 
\epsilon_{abcde_1 e_2 e_3}\, \dd x^{e_1 e_2 e_3}
\\[5pt]
&= - \frac{5}{12}\, \lambda^{abc}\, g_{c[a}\, \psi_{b e_1 e_2 e_3]}\, \dd x^{e_1 e_2 e_3}
= - \frac{1}{4}\, \lambda^{ab}{}_{e_1}\, \psi_{e_2 e_3 ab}\, \dd x^{e_1 e_2 e_3} 
\end{align*}
where we have used equation \eqref{eq:g2phiep} in the first line.
Representing $\lambda$ in terms of a symmetric traceless matrix $h$ as
\[ \lambda = \frac{1}{2}\, h^a_b\, \varphi_{acd}\, \dd x^{bcd}~,\]
we have 
\begin{align*}
*\left(- \frac{1}{4}\, \lambda_{ab}{}^e\,  \varphi_{ecd}\, \dd x^{abcd}\right)
&= - \frac{3}{4}\,  h_{[a}^c\, \varphi_{b e_1]c}\,\psi^{ab}{}_{e_2 e_3}\, \dd x^{e_1 e_2 e_3}
\\[5pt]
&=  - \frac{1}{4}\,  (h_{e_1}^c\, \varphi_{abc}
+ 2\, h_a^c\, \varphi_{be_1c})
\,\psi^{ab}{}_{e_2 e_3}\, \dd x^{e_1 e_2 e_3}
\\[5pt]
&=  - \frac{1}{4}\,  ( 4\, h_{e_1}^c\, \varphi_{c e_2 e_3}
- 12\, h_c^a\, g _{e_1 d}\,\delta_{[a}{}^{[d}  \, \varphi^{c]}{}_{e_2 e_3]})
\, \dd x^{e_1 e_2 e_3}
\\[5pt]
&=  - 2\, \lambda 
+ h_c^a\, g _{e_1 d}\,(\delta_{a}{}^{[d}  \, \varphi^{c]}{}_{e_2 e_3}
+ 2\, \delta_{e_2}{}^{[d}  \, \varphi^{c]}{}_{e_3 a})
\, \dd x^{e_1 e_2 e_3}
\\[5pt]
&=  - 2\, \lambda + 3\, \lambda  = \lambda~,
\end{align*}
where we have used identities \eqref{eq:g2ident4} and \eqref{eq:g2ident5}.
The second identity follows easily by showing that 
\[ *(\lambda_{abf}\, \psi^f{}_{cde}\, \dd x^{abcde}) = 0~.\]

\end{proof}

\begin{lemma}
Let $\alpha\in\Lambda^2(Y, TY)$. Then
\begin{align}
\pi_7(i_\alpha(\psi)) &= - (\pi_{14}(\alpha^a))_{ab}\, \dd x^b\wedge\psi~,\label{eq:5insert7}
\\
\pi_{14}(i_\alpha(\psi)) &= i_{\pi_7(\alpha)}(\psi)~.\label{eq:5insert14}
\end{align}
\end{lemma}
\begin{proof}
\[
i_\alpha(\psi) = \alpha^a\wedge\psi_a = \pi_7(\alpha^a)\wedge\psi_a + \pi_{14}(\alpha^a)\wedge\psi_a
= \pi_7(\alpha^a)\wedge\psi_a - (\pi_{14}(\alpha^a))_{ab}\, \dd x^b\wedge\psi~,
\]
where we have used  identity \eqref{eq:idtwo14}.   Contracting the first term with $\psi$ we find
\[ \psi\lrcorner(\pi_7(\alpha^a)\wedge\psi_a) = 0~,\]
hence equations \eqref{eq:5insert7} and \eqref{eq:5insert14} follow.
\end{proof}
\vskip10pt

\begin{lemma}
Let $\alpha$ and $\beta$ be two forms in $\Lambda_{14}^2$.  Then
\begin{equation}
\gamma = \frac{1}{2}\, \alpha^a\wedge\beta_{ab}\, \dd x^b \in \Lambda_{14}^2(Y)  ~,
\label{eq:i14to14}
\end{equation}
where
\[ \alpha^a = g^{ab}\, \alpha_{bc}\, \dd x^c~.\]
\end{lemma}
\begin{proof}
To prove equation \eqref{eq:i14to14}, we prove that $\gamma\lrcorner\varphi = 0$.
We have 
\begin{align*}
\gamma\lrcorner\varphi &= 
\frac{1}{2}\, 
\alpha_b^a\,\beta_{ac}\,\varphi^{bc}{}_d\, \dd x^d
= - \beta^{ac}\, \alpha_{ba} \,\varphi^b{}_{cd}\, \dd x^d
\\[5pt]
&= - \frac{1}{2}\,  \beta^{ac}\, (3\, \alpha_{b[a} \,\varphi^b{}_{cd]}
- \alpha_{bc} \,\varphi^b{}_{da} - \alpha_{bd} \,\varphi^b{}_{ac})\, \dd x^d~.
\end{align*}
The last term vanishes as $\beta\lrcorner\varphi = 0$.  The first term also vanishes by lemma 4 of \cite{delaOssa:2016ivz}. 
Then
\begin{equation*}
\gamma\lrcorner\varphi= \frac{1}{2}\,  \beta^{ac}\,  \alpha_{bc} \,\varphi^b{}_{da} \, \dd x^d
=  \frac{1}{2}\,  \alpha^c_b \,\beta_{ca}\,  \varphi^{ab}{}_d \, \dd x^d  = - \gamma\lrcorner\varphi~,
\end{equation*}
and therefore $\gamma\lrcorner\varphi=0$.
\end{proof}
\vskip10pt

\newpage
\section{Curvature identities}\label{sec:appcurvs}

In this appendix we prove curvature identities that hold for the connections on manifolds with $G_2$ structure. We focus on two connections: the $G_2$ holonomy connection $\nabla$ with totally antisymmetric torsion $H$, defined in section \ref{sec:inttor}, and the connection $\dth$, defined in section \ref{sec:covder}. We will, in particular, show that $\dth$ is not an instanton connection.  

Let $Y$ be a Riemannian manifold and $\nabla$ a connection on $Y$ with connection symbols $\Gamma$ and corresponding spin connetion $\Omega$. The curvature $R(\Gamma)$ of the connection $\nabla$ is defined by
\begin{align*}
R(\Gamma)_a{}^b &=  \frac{1}{2}\, (R(\Gamma)_a{}^b)_{cd}\, \dd x^{cd}
= (\partial_{c}\Gamma_{da}{}^b
+ \Gamma_{ce}{}^b\, \Gamma_{da}{}^e)\, \dd x^{cd}
\\[5pt]
&= - (\partial_c\Omega_{d\alpha\beta} 
+ \Omega_{c\alpha\gamma}\, \Omega_d{}^\gamma{}_\beta)
\, e_a{}^\alpha\, e^{b\beta}
~.
\end{align*}
If $\eta$ is a spinor on $Y$ we have
\[
 [\nabla_a, \nabla_b]\, \eta = - \frac{1}{4}\, (R(\Gamma)_{cd})_{ab}\, \gamma^{cd}\, \eta
- T_{ab}{}^c\, \nabla_c\eta~,
\]
where $T$ is the torsion of the connection and $\gamma^a$ are the $\gamma$ matrices generating the Clifford algebra of $Spin(7)$.

\begin{proposition}
Let  $Y$ be a Riemannian manifold, and let $\nabla$  be a metric connection on $Y$ with connection symbols 
\[ 
\Gamma_{ab}{}^c= \Gamma^{LC}_{\,ab}{}^{\,c} + A_{ab}{}^c~.
\]
 Then
\[
R(\Gamma)_a{}^b - R(\Gamma^{LC}) _a{}^b =
(  \nabla^{LC}_{\, c}\, A_{da}{}^b + A_{ce}{}^b\, A_{da}{}^e) \, \dd x^{cd}~.
\]
\end{proposition}

\begin{proof}
Consider first the curvature of the connection $\nabla$ with 
connection symbols $\Gamma$, which can be written   as
\[ \Gamma_{ab}{}^c = \Gamma^{LC}_{\, ab}{}^{\, c} + A_{ab}{}^c~.
\]
Then
\begin{align*}
R(\Gamma)_a{}^b - R(\Gamma^{LC})_a{}^b 
&=  
(\partial_{c}A_{da}{}^{b}
+ \Gamma^{LC}_{\, ce}{}^{\,b}\, A_{da}{}^e
+ A_{ce}{}^b\, \Gamma^{LC}_{\, da}{}^{\, e}
+ A_{ce}{}^b\, A_{da}{}^e)\, \dd x^{cd}
\\[5pt]
&=  
(\nabla^{LC}_{\, c}\, A_{da}{}^{b}
+ A_{ce}{}^b\, A_{da}{}^e)\, \dd x^{cd}~.
\end{align*}

\end{proof}

Suppose now that $Y$ admits a well defined nowhere vanishing Majorana spinor $\eta$, and therefore has a $G_2$ structure determined by
\[ \varphi_{abd} = -i\, \eta^\dagger\, \gamma_{abc}\, \eta~.\]
Suppose
\[ \nabla_a\,\eta = 0~,\] 
where $\nabla$ is a connection with $G_2$ holonomy on $Y$.
Then the curvarture of the connection $\nabla$ satisfies
\[
(R(\Gamma)_{ab})_{cd} \, \varphi^{ab}{}_e = 0~.
\]
Thus, $\nabla$ is an instanton connection on $Y$. In particular, this holds for the unique $G_2$ holonomy connection with totally antisymmetric torsion $A_{abc} = \frac{1}{2}H_{abc}$. We will restrict to this connection in the following.

On manifolds with a  $G_2$ structure we have defined a connection $\dth$ in terms of a $G_2$ compatible connection $\Gamma$ which acts on forms with values in $TY$ by
\[ 
\dth\Delta^a = \dd\Delta^a + \theta_b{}^a\wedge\Delta^b~,
\]
where $\theta_a{}^b = \Gamma_{ac}{}^b\, \dd x^c$. Note that this connection
is not compatible with the $G_2$ structure and that it is not necessarily metric either.
The curvature $R(\theta)$ of this connection is
\[ 
R(\theta)_a{}^b = \dd\theta_a{}^b + \theta_c{}^b\wedge\theta_a{}^c
= (\partial_c\Gamma_{ad}{}^b + \Gamma_{ec}{}^b\, \Gamma_{ad}{}^e)\, 
\dd x^{cd}
~.\]

\begin{proposition}
Let  $Y$ be a manifold with a $G_2$ structure determined by $\varphi$.  Let $\nabla$  be a metric connection compatible with the $G_2$ structure (that is $\nabla\varphi$ = 0) and with connection symbols 
\[ \Gamma_{ab}{}^c = \Gamma^{LC}_{\, ab}{}^{\, c} + \frac{1}{2} H_{ab}{}^c~.\]
  Then the curvature of the connection $\dth$ satisfies
  \[
R(\theta)_a{}^b - R(\Gamma^{LC}) _a{}^b = 
\frac{1}{4}(2 \nabla^{LC}_c\, H_{ad}{}^b 
+ H_{ec}{}^b\, H_{ad}{}^e)\, \dd x^{cd}~.
\]

\end{proposition}

\begin{proof}
We have, from the definitions of the curvatures of the connections,
\begin{align*}
R(\theta)_a{}^b -  R(\Gamma^{LC})_a{}^b
&=  
\frac{1}{4} (2\partial_c\, H_{ad}{}^b  
+2 \Gamma^{LC}_{\,ec}{}^{\,b}\, H_{ad}{}^e
+2 H_{ec}{}^b\, \Gamma^{LC}_{\, ad}{}^{\,e}
+ H_{ec}{}^b\, H_{ad}{}^e
)\,\dd x^{cd}
\\[5pt]
&= \frac{1}{2} (2 \nabla^{LC}_c\, H_{ad}{}^b  
+ H_{ec}{}^b\, H_{ad}{}^e
)\,\dd x^{cd}~.
\end{align*}

\end{proof}

\begin{proposition}
If the connection $\Gamma$ has totally antisymmetric torsion, the curvatures of the connection $\nabla$ and $\dth$ are related by the identity
\[
(R(\Gamma)_{cd})_{ab} - (R(\theta)_{ab})_{cd} =
\frac{1}{2}\, (\dd H)_{abcd}~.
\]
\end{proposition}

\begin{proof}
Recalling that
\[ (R(\Gamma^{LC}) _{cd})_{ab} = (R(\Gamma^{LC}) _{ab})_{cd}~,\]
we find
\begin{align*}
(R(\Gamma)_{cd})_{ab} - (R(\theta)_{ab})_{cd} &= 
\, (\nabla^{LC}_{\, [a}\, H_{b]cd} - \nabla^{LC}_{[c}\, H_{|a|d]b})
\\[5pt]
&\qquad +\frac{1}{2}\left( H_{aed}\, H_{bc}{}^e - H_{bed}\, H_{ac}{}^e
- H_{ecb}\, H_{ad}{}^e + H_{edb}\, H_{ac}{}^e\right)\\
&= 2\, \nabla^{LC}_{\, [a}\, H_{bcd]}
= 2\, \partial_{[a}\, H_{bcd]}
\\[5pt]
&= \frac{1}{2}\, (\dd H)_{abcd}
~.
\end{align*}
\end{proof}

\begin{proposition}
\label{prop:thetanotinst}
The Bianchi identity of the anomaly cancelation condition implies
\[
R(\theta)_{ab}\,\lrcorner\varphi 
= \frac{\alpha'}{8}\, \big(\tr (F_{ac}\, F_{bd})
- \tr (\tilde R_{ac}\, \tilde R_{bd})\big)\, \varphi^{cd}{}_e\, \dd x^e~,
\]
where $\tilde R$ is the curvature of an instanton connection on $TY$.
\end{proposition}

\begin{proof}
Recall that 
\[ 
(R(\Gamma)_{cd})_{ab}\, \varphi^{cde} = 0~.\]
Then, by the previous proposition
\begin{equation*}
R(\theta)_{ab}\,\lrcorner\varphi 
= - \frac{1}{4}\, (\dd H)_{cdab}\, \varphi^{cd}{}_e\, \dd x^e~.
\end{equation*}
By the Bianchi identity
\begin{align*}
(\dd H)_{cdab}\, \varphi^{cd}{}_e
&= \frac{\alpha'}{4}\, 3!\, \big(\tr (F_{[cd}\, F_{ab]})
- \tr (\tilde R_{[cd}\, \tilde R_{ab]})\big)\, \varphi^{cd}{}_e
\\[5pt]
&= - \alpha'\, \big(\tr (F_{ac}\, F_{bd})
- \tr (\tilde R_{ac}\, \tilde R_{bd})\big)\, \varphi^{cd}{}_e
\end{align*}
where in the last line we have used the fact that $F\lrcorner\varphi= 0$ and
$\tilde R\lrcorner\varphi = 0$. 
Therefore
\[
R(\theta)_{ab}\,\lrcorner\varphi 
= \frac{\alpha'}{4}\, \big(\tr (F_{ac}\, F_{bd})
- \tr (\tilde R_{ac}\, \tilde R_{bd})\big)\, \varphi^{cd}{}_e\, \dd x^e~.
\]

\end{proof}

Note that this means that the connection $\theta$ is not an instanton. To expand on this fact, note that the right hand side of this equation is zero if the $F$ equals $\tilde{R}$. In the string compactification literature this is known as the standard embedding of the gauge bundle in the tangent bundle, and leads to a vanishing flux $H$. Thus, we have reduced to a $G_2$ holonomy compactification, where $\dth$ is in fact identical with the Levi-Civita connection. The reader is referred to \cite{delaOssa:2016ivz} for more details on this case.

\subsection*{Curvature and Covariant Derivatives of Torsion Classes}

We now collect some useful identities between the covariant derivatives of the torsion classes and the curvature $R(\theta)$.

\begin{proposition}
Let $Y$ be a manifold with a $G_2$ structure (not necessarily integrable), and let $\nabla$ be a metric connection compatible with this $G_2$ structure, that is
\[ \nabla\varphi = 0~, \qquad \nabla\psi = 0~.\]
Then, 
\begin{align}
(\nabla_a\tau_0)\, \psi + 3\, (\nabla_a\tau_1)\wedge\varphi
+ \nabla_a *\tau_3
&= \frac{1}{2}\, R(\theta)_a{}^b\wedge\varphi_{cdb}\, \dd x^{cd} ~,
\label{eq:covtauone}\\[7pt]
4 \, (\nabla_a\tau_1)\wedge\psi - \nabla_a\tau_2\wedge\varphi
&= -  \frac{1}{3!}\, R(\theta)_a{}^b\wedge\psi_{cdeb}\, \dd x^{cde}~,
\label{eq:covtautwo}
\end{align}
where
\[ R(\theta)_a{}^b = \dd\theta_a{}^b + \theta_c{}^b\wedge\theta_a{}^c~,\]
is the curvature of the connection $\dth$. 

\end{proposition}

\begin{proof}
We begin by taking the covariant derivative of the integrability equations \eqref{eq:Intphi} and \eqref{eq:Intpsi}.  We find
\begin{align*}
\nabla_a\dd\varphi 
&= (\nabla_a\tau_0)\, \psi + 3\, (\nabla_a\tau_1)\wedge\varphi
+ \nabla_a *\tau_3~,
\\[7pt]
\nabla_a\dd\psi &= 
4 \, (\nabla_a\tau_1)\wedge\psi + \nabla_a*\tau_2~.
\end{align*}
For the first equation, lemma 8 of \cite{delaOssa:2016ivz} together with a bit of algebra leads to the equation 
\begin{align*}
\dd\partial_a\varphi &= \frac{1}{2}\, R_a{}^b(\theta)\wedge\varphi_{cdb}\, \dd x^{cd}
- \frac{1}{3!}\, \theta_a{}^b\wedge(\dd\varphi)_{cdeb}\, \dd x^{cde}~.
\end{align*}
Then equation \eqref{eq:covtauone} follows from this together with
\[ \nabla_a\dd\varphi = \partial_a\dd\varphi + \frac{1}{3!}\, \theta_a{}^b\wedge(\dd\varphi)_{cdeb}\, \dd x^{cde}~.\]
The proof of equation \eqref{eq:covtautwo} is analogous.

\end{proof}

Note in particular that from \eqref{eq:covtauone}-\eqref{eq:covtautwo} we can derive the covariant derivatives of the torsion classes soely in terms of the curvature $R(\theta)$. Note also that if the $G_2$ structure is integrable, equation \eqref{eq:covtautwo} implies that there is a constraint on the curvature of the connection $\theta$
\[ \pi_{14}\left(R(\theta)_a{}^b\wedge\psi_{b}\right) = 0~. \]
Then by equation \eqref{eq:5insert14}, we find that the curvature of the connection $\theta$ must satisfy
\begin{equation}
 \check R_a{}^b(\theta)\wedge\psi_{b} = 0~.
 \label{eq:curvintG2}
\end{equation}


\section{Heterotic supergravity and equations of motion}
\label{app:sugra}
In this appendix we briefly review heterotic supergravity, the Killing spinor equations and comment on the corresponding equations of motion. Recall first the bosonic part of the action \cite{Bergshoeff1989439}
\begin{align}
	S_B = \int \sqrt{-g} e^{-2\phi} d^{10}x \Big[ &{\cal R} + 4(\dd \phi)^2 - \frac{1}{12}{\cal H}_{\mu\nu\rho}{\cal H}^{\mu\nu\rho}\notag\\
		&- \tfrac{\alpha'}{8} \tr F_{\mu \nu} F^{\mu \nu}
		+ \tfrac{\alpha'}{8} \tr {R(\tilde\theta)}_{\mu \nu} \tr {R(\tilde\theta)}^{\mu \nu} \Big]+{\cal O}(\alpha'^2)\:,
	\label{eq:action10d}
\end{align}
where $\{\mu,\nu,..\}$ denote ten dimensional indecies, $\cal R$ is the Ricci scalar, $\phi$ is the dilaton, and ${\cal H}$ is the Neveu-Schwarz three-form flux given by
\begin{equation}
\label{eq:defH}
 {\cal H} = \dd B + \frac{\alpha'}{4}\, ({\cal CS}[A] - {\cal CS}[\tilde\theta]) \; ,
\end{equation}
where $B$ is the Kalb-Ramond two-form. Under gauge transformations $\{\epsilon_1,\epsilon_2\}$ of $\{A,\tilde\theta\}$ respectively, the $B$ field is required to transform as
\begin{equation}
\delta B=-\frac{\alpha'}{4}\Big({\rm tr}\:(\dd A\epsilon_1)-{\rm tr}\:(\dd\tilde\theta\epsilon_2)\Big)\:,
\end{equation}
in order for ${\cal H}$ to remain gauge-invariant \cite{green1984anomaly}.

The supersymmetry conditions read \cite{green1987superstring2, Bergshoeff1989439}
\begin{align}
	\nabla_\mu\epsilon=(\nabla^{LC}_\mu + \frac{1}{8} {\cal H}_{\mu\nu\lambda} \gamma^{\nu\lambda})\, \epsilon&=0+{\cal O}(\alpha'^2)\notag\\
	({\slashed\nabla}^{LC} + \frac{1}{4} \slashed {\cal H} - \slashed\partial \phi )\,\epsilon&=0+{\cal O}(\alpha'^2)\notag\\
	\slashed{F}\, \epsilon &= 0+{\cal O}(\alpha')\:,
\label{eq:susytrans}
\end{align}
where $\psi_\mu$ is the gravitino, $\rho$ is the modified dilatino and $\chi$ is the gaugino. Here the last condition is only required at zeroth order since the gauge field only appears at first order in the theory. These supersymmetry conditions are accurate, provided we also choose the connection $\tilde\theta$ to satisfy an instanton condition \cite{Ivanov:2009rh, delaOssa:2014msa}
\begin{equation}
\slashed R(\tilde\theta)\,\epsilon=0+{\cal O}(\alpha')\:.
\end{equation}
In the above, we have defined for a $p$-form $\alpha$ 
\begin{equation}
\slashed\alpha=\frac{1}{p!}\alpha_{\mu_1..\mu_p}\gamma^{\mu_1..\mu_p}\:.
\end{equation}

\subsection*{$G_2$ Reductions and Equations of Motion}
We wish to reduce the supersymmetry transformations \eqref{eq:susytrans} on spacetimes of the form
\begin{equation}
\label{eq:7dreduction}
M_{10}=M_3\times Y\:,
\end{equation}
where $M_3$ is maximally symmetric. We suppose that $Y$ admits a well defined nowhere vanishing Majorana spinor $\eta$, and therefore has a $G_2$ structure determined by
\[ \varphi_{abd} = -i\, \eta^\dagger\, \gamma_{abc}\, \eta\:,\;\;\;\psi=*\varphi~.\]
Using this, one arrives at the supersymmetry conditions \cite{Gauntlett:2001ur, 2001math......2142F, Firedrich:2003, Lukas:2010mf, Gray:2012md}
\begin{align}
\dd\varphi &= 2\, \dd\phi\wedge\varphi -  * H - f\, \psi~,\label{eq:7dsusy1}\\
\dd\psi &= 2\, \dd\phi\wedge\psi~,\label{eq:7dsusy2}\\
\textstyle{\frac{1}{2}}\,  * f &= H\wedge\psi~,\label{eq:7dsusy3}\\
0&=F\wedge\psi\:,\label{eq:7dsusy3}
\end{align}
where now the three-form $H$ and the constant $f$ are components of the ten-dimensional flux ${\cal H}$, which lie along $Y$ and the three-dimensional, maximally symmetric world-volume, respectively.\footnote{The flux component $f$ determines the cosmological constant of the three-dimensional spacetime through the Einstein equation of motion. A zero/non-zero $f$ gives Minkowski/AdS spacetimes respectively.} We have also restricted the bundle to the internal geometry. Generic solutions to these equations imply that $Y$ has an integrable $G_2$ structure where $\tau_1$ is exact.

It can be shown that for compactifications of the form
\begin{equation}
M_{10}=M_d\times X_{10-d}\:,
\end{equation}
where $M_d$ is maximally symmetric, provided the flux equation of motion is satisfied, the supersymmetry equations will also imply the equations of motion \cite{Gauntlett:2002sc}. Note that the authors of \cite{Gauntlett:2002sc} assume $M_d$ to be Minkowski, but the generalisation to AdS is straight forward. In our case, the flux equation of motion on the spacetime \eqref{eq:7dreduction} reduces to
\begin{equation}
\dd(e^{-2\phi}*H)=0\:,
\end{equation}
which can easily be checked is satisfied from \eqref{eq:7dsusy1}-\eqref{eq:7dsusy2}.

\subsection*{Comments on $\tilde\theta$ and Field Redefinitions.}
Let us make a couple of comments concerning the connection $\tilde\theta$ appearing in both the action and the definition of $H$ equation \eqref{eq:defH}, often referred to as the anomaly cancellation condition. In deriving the heterotic action, Bergshoeff and de Roo \cite{Bergshoeff1989439} used the fact that $(\hat\theta,\psi^+)$ transforms as an $SO(9,1)$ Yang-Mills supermultiplet modulo $\alpha'$ corrections. Here $\theta$ is the connection whose connection symbols read
\begin{equation}
{\theta_{\mu\nu}}^\rho = {\Gamma_{\nu\mu}}^\rho\:,
\end{equation}
where the $\Gamma$'s denote the connection symbols of $\nabla$. The connection $\hat\theta$ then denotes an appropriate fermionic correction to $\theta$, while $\psi^+$ is the supercovariant gravitino curvature. Modulo ${\cal O}(\alpha'^2)$-corrections, they could then construct a supersymmetric theory with curvature squared corrections, simply by adding the appropriate $SO(9,1)$-Yang-Mills action to the theory. The resulting bosonic action then uses $\theta$ rather than $\tilde\theta$.

In the bulk of the paper we have replaced $\theta$ in with a more general connection $\tilde\theta$ in the appropriate places. Ambiguities surrounding the connection $\tilde\theta$ have been discussed extensively in the literature before \cite{Hull1986187, Sen1986289, Hull198651, Hull1986357, 0264-9381-4-6-027, Becker:2009df, Ivanov:2009rh, Melnikov:2012cv, Melnikov:2012nm, Melnikov:2014ywa, delaOssa:2014msa}. In particular, it has been argued that deforming this connection can equivalently be interpreted as a field redefinition, though care most be taken when performing such redefinitions as they in general also lead to corrections to the supersymmetry transformations and equations of motion. In particular, we argued in \cite{delaOssa:2014msa} that in order to preserve \eqref{eq:susytrans} as the correct supersymmetry conditions, one must choose $\tilde\theta$ to satisfy the instanton condition modulo $\alpha'$-corrections.\footnote{It should be noted that the arguments in \cite{delaOssa:2014msa} where for the most part restricted to the Strominger-Hull system, although we expect them to hold true for the heterotic $G_2$ system as well.} Note that although $\theta$ satisfies the instanton condition to zeroth order in $\alpha'$, it generically fails to do so once higher order corrections are included. Indeed, this was crucial for the mathematical structure presented in this paper.


\newpage
\bibliographystyle{JHEP}

\bibliography{bibliography}
\end{document}